\documentclass[10pt,letterpaper]{IEEEtran}

\usepackage{ucs}
\usepackage[utf8x]{inputenc}
\usepackage{cite}

\usepackage{tikz} 
\usetikzlibrary{shapes,arrows,calc,intersections,positioning,arrows.meta}
\usepackage{pgfplots}
\usepgfplotslibrary{statistics}
\usepackage{pgfplotstable}
\usepackage{tumcolor}

\usepackage[cmex10]{amsmath}
\usepackage{amssymb}
\usepackage{amsthm}
\usepackage{noindentafter}
\usepackage{bm}
\usepackage{fixmath}
\usepackage{mathtools}
\usepackage{amsbsy}
\usepackage{algorithmicx}
\usepackage{algpseudocode}
\usepackage{algorithm}
\usepackage{booktabs}
\usepackage{siunitx}

\usepackage{array}

\pgfplotsset{compat=1.13}

\usepackage[caption=false,font=footnotesize]{subfig}

\usepackage{relsize}
\usepackage{dane}

\begin{document}

\title{Learning the MMSE Channel Estimator}

\author{David Neumann,
        Thomas~Wiese,~\IEEEmembership{Student Member,~IEEE,}
        and~Wolfgang~Utschick,~\IEEEmembership{Senior Member,~IEEE}% <-this % stops a space
        \thanks{Copyright (c) 2018 IEEE. Personal use of this
material is permitted. However, permission to use this material for
any other purposes must be obtained from the IEEE by sending a request
to pubs-permissions@ieee.org.}
        \thanks{The authors are with the Professur f\"ur Methoden der Signalverarbeitung, Technische Universit\"at M\"unchen, 80290 M\"unchen, Germany (email: \{d.neumann, thomas.wiese, utschick\}@tum.de).}
\thanks{A shorter version of this paper was presented at the 21st International ITG Workshop on Smart Antennas (WSA), Berlin, Germany, 2017}%
}

\markboth{}%
{}

\maketitle

\begin{abstract}
We present a method for estimating conditionally Gaussian random vectors with random covariance matrices, which uses techniques from the field of machine learning.
Such models are typical in communication systems, where the covariance matrix of the channel vector depends on random parameters, e.g., angles of propagation paths.
If the covariance matrices exhibit certain Toeplitz and shift-invariance structures, the complexity of the MMSE channel estimator can be reduced to $\order(M\log M)$ floating point operations, where $M$ is the channel dimension.
While in the absence of structure the complexity is much higher, we obtain a similarly efficient (but suboptimal) estimator by using the MMSE estimator of the structured model as a blueprint for the architecture of a neural network.
This network learns the MMSE estimator for the unstructured model, but only within the given class of estimators that contains the MMSE estimator for the structured model.
Numerical simulations with typical spatial channel models demonstrate the generalization properties of the chosen class of estimators to realistic channel models.
\end{abstract}


\IEEEpeerreviewmaketitle

\section{Introduction}

Accurate channel estimation is a major challenge in the next
generation of wireless communication networks, e.g., in cellular
massive
MIMO~\cite{rusek_scaling_2013,marzetta_noncooperative_2010}
or millimeter-wave~\cite{Mota14,Heath14} networks.
In setups with many antennas and low signal to noise ratios (SNRs), errors in the channel estimates are particularly devastating, because the array gain cannot be fully realized.
Since a large array gain is essential in such setups, there is currently a lot of research going on concerning the modeling and verification of massive MIMO and/or millimeter wave channels~\cite{Rappaport16,Sun16} and the question how these models can aid channel estimation~\cite{yin_coordinated_2013}.

For complicated stochastic models, the minimum mean squared error (MMSE) estimates of the channel cannot be calculated in closed form.
A common strategy to obtain computable estimators is to restrict the estimator to a certain class of functions and then find the best estimator in that class.
For example, we could restrict the estimator to the class of linear operators.
The linear MMSE (LMMSE) estimator is then represented by the optimal linear operator, i.e., the linear operator that minimizes the mean squared error (MSE).
In some special cases, the matrix that represents the optimal linear estimator can be calculated in closed form; in other cases, it has to be calculated numerically.
We know that the LMMSE estimator is the MMSE estimator for jointly Gaussian distributed random variables.
Nonetheless, it is often used in non-Gaussian settings and performs well for all kinds of distributions that are not too different from a Gaussian.

In the same spirit, we present a class of low-complexity channel estimators, which contain a convolutional neural network (CNN) as their core component.
These CNN-estimators are composed of convolutions and some simple nonlinear operations.
The CNN-MMSE estimator is then the CNN-estimator with optimal convolution kernels such that the resulting estimator minimizes the MSE.
These optimal kernels have to be calculated numerically, and this procedure is called \emph{learning}.
Just as the LMMSE estimator is optimal for jointly Gaussian random variables, the CNN-MMSE estimator is optimal for a specific idealized channel model (essentially a single-path model as described by the ETSI 3rd Generation Partnership Project (3GPP)~\cite{3gpp}).
In numerical simulations, we find that the CNN-MMSE estimator works fine for the channel models proposed by the 3GPP, even though these violate the assumptions under which the CNN-MMSE estimator is optimal.

Once we have learned the CNN-MMSE estimator from real or simulated channel realizations, the computational complexity required to calculate a channel estimate is only $\order(M\log M)$ floating point operations (FLOPS).
Despite this low complexity, the performance of the CNN-MMSE estimator does not trail far behind that of the unrestricted MMSE estimator, which is very complex to compute.
Since the learning procedure is performed off-line, it does not add to the complexity of the estimator.

One assumption of the idealized channel model mentioned above is that the covariance matrices have Toeplitz structure.
This assumption has also motivated other researchers to propose estimators that exploit this structure.
For example, in~\cite{Heath14} and \cite{Rao14}, methods from the area of compressive sensing are used to approximate the channel vector as a linear combination of $k$ steering vectors where $k$ is much smaller than the number of antennas $M$.
With these methods, a complexity of $\order(M\log M)$ floating point operations can be achieved if efficient implementations are used.
Although of similar complexity, the proposed CNN-MMSE estimator significantly outperforms the compressive-sensing-based estimators in our simulations.

In~\cite{Haghighatshoar17}, the maximum likelihood estimator of the channel covariance matrix within the class of all positive semi-definite Toeplitz matrices is constructed.
This estimated covariance matrix is then used to estimate the actual channel.
However, even the low-complexity version of this covariance matrix estimator relies on the solution of a convex program with $M$ variables, i.e., its complexity is polynomial in the number of antennas.
There also exists previous work on learning-based channel estimation~\cite{omri2010channel,zhang2007mimo,prasad_joint_2014,zhou_channel_2003}, but with completely different focus in terms of system model and estimator design. 

In summary, our main contributions are the following:
\begin{itemize}
    \item We derive the MMSE channel estimator for conditionally normal channel models, i.e., the channel is normally distributed given a set of parameters, which are also modelled as random variables.
    \item We show how the complexity of the MMSE estimator can be reduced to $\order(M\log M)$ if the channel covariance matrices are Toeplitz and have a shift-invariance structure.
    \item We use the structure of this MMSE estimator to define the CNN estimators, which have $\order(M\log M)$ complexity.
    \item We describe how the variables of the neural network can be optimized/learned for a general channel model using stochastic gradient methods.
    \item We introduce a hierarchical learning algorithm, which helps to avoid local optima during the learning procedure.
\end{itemize}

\subsection{Notation}

The transpose and conjugate transpose of $\vX$ are denoted by $\vX\tp$ and $\vX\he$, respectively.
The trace of a square matrix $\vX$ is denoted by $\trace(\vX)$ and the Frobenius norm of $\vX$ is denoted by $\norm{\vX}_F$.
Two matrices $\vA, \vB \in \C^{M\times M}$ are asymptotically equivalent, which we denote by $\vA\asymp\vB$, if
\begin{align}
    \lim_{M\rightarrow \infty} \norm{\vA - \vB}_F^2/M = 0.
\end{align}
We write $\exp(\vx)$ and $\abs{\vx}^2$ to denote element-wise application of $\exp(\cdot)$ and $\abs{\cdot}^2$ to the elements of $\vx$.
The $k$th entry of the vector $\vx$ is denoted by $[\vx]_k$; similarly, for a matrix $\vX$, we write $[\vX]_{mn}$ to denote the entry in the $m$th row and $n$th column.
The circular convolution of two vectors $\va,\vb \in \C^M$ is denoted by $\va \ast \vb \in \C^M$.
Finally, $\diag(\vx)$ denotes the square matrix with the entries of the vector $\vx$ on its diagonal and $\vec(\vX)$ is the vector obtained by stacking all columns of the matrix $\vX$ into a single vector.

\section{Conditionally Normal Channels}

We consider a base station with $M$ antennas, which receives uplink training signals from a single-antenna terminal.
We assume a frequency-flat, block-fading channel, i.e., we get independent observations in each coherence interval.
After correlating the received training signals with the pilot sequence transmitted by the mobile terminal, we get observations of the form
\begin{equation}
    \vy_t = \vh_t + \vz_t,\quad t=1,\ldots,T
\end{equation}
with the channel vectors $\vh_t$ and additive Gaussian noise $\vz_t \sim \CN(\zeros,\cz)$.
For the major part of this work, we assume that the noise covariance is a scaled identity $\cz = \sigma^2 \id$ with \emph{known} $\sigma^2$.
The channel vectors are assumed to be conditionally Gaussian distributed given a set of parameters $\vdelta$, i.e., $\vh_t\vert\vdelta \sim \CN(\zeros,\vC_\vdelta)$.
In contrast to the fast-fading channel vectors, the covariance matrix $\vC_\vdelta$ is assumed to be constant over the $T$ channel coherence intervals.
That is, $T$ denotes the coherence interval of the covariance matrix in number of channel coherence intervals.

The parameters, which describe, for example, angles of propagation paths, are also considered as random variables with distribution $\vdelta \sim p(\vdelta)$, which is known.
In summary, we have
\begin{equation}\label{eq:modelbayes1}
    \vy_t \vert \vh_t \sim \CN(\vh_t, \cz)
\end{equation}
with \emph{known} noise covariance matrix $\cz$ and hierarchical prior
\begin{equation}\label{eq:modelbayes2}
    \vh_t \vert \vdelta \sim \CN(\zeros,\vC_\vdelta)\,,\quad \vdelta \sim p(\vdelta)\,.
\end{equation}

\textbf{Example.}
Conditionally normal channels appear in typical channel models for communication scenarios, e.g., in those defined by the 3GPP for cellular networks~\cite{3gpp}.
There, the covariance matrices are of the form
\begin{equation}\label{eq:covmodel}
    \vC_\vdelta = \int_{-\pi}^{\pi} g(\theta;\vdelta)\va(\theta)\va(\theta)\he d\theta\,,
\end{equation}
where $g(\theta;\vdelta)\geq 0$ is a power density function
corresponding to the parameters $\vdelta$ and where $\va(\theta)$ denotes the array manifold vector of the antenna array at the base station for an angle $\theta$.
As an example, in the 3GPP urban micro and urban macro scenarios, $g(\theta;\vdelta)$
is a superposition of several scaled probability density functions of a
Laplace-distributed random variable with standard deviations $2^\circ$
and $5^\circ$, respectively.
The Laplace density models the scattering of the received power around the center of the propagation path.
Its standard deviation is denoted as the per-path angular spread.

\section{MMSE Channel Estimation}\label{sec:mmse}

Our goal is to estimate $\vh_t$ for each $t$ given all observations $\vY = \big[\vy_1,\ldots,\vy_T\big]$ and knowledge of the model~\eqref{eq:modelbayes1}, \eqref{eq:modelbayes2}.
For a fixed parameter $\vdelta$, the MMSE estimator can be given analytically, since conditioned on $\vdelta$, the observation $\vy_t$ is jointly Gaussian distributed with the channel vector $\vh_t$.
Also, given the parameters $\vdelta$, the observations of different coherence intervals are independent.
The conditional MMSE estimate of the channel vector $\vh_t$ is~\cite{yin_coordinated_2013,kailath_linear_1980}
\begin{equation}\label{eq:cond_mmse_estimator1}
  \expec[\vh_t \vert \vY, \vdelta] = \vW_\vdelta \vy_t
\end{equation}
with
\begin{equation}\label{eq:cond_mmse_estimator2}
    \vW_\vdelta = \vC_\vdelta (\vC_\vdelta + \cz)\inv.
\end{equation}
Given the parameters $\vdelta$, the estimate of $\vh_t$ only depends on $\vy_t$.
Since the parameters $\vdelta$ and, thus, the covariance matrix $\vC_\vdelta$ are unknown random variables, the MMSE estimator for our system model is given by
\begin{align}
    \hest_t
    &= \expec[\vh_t \,\vert\, \vY] \\
    &= \expec[\,\expec[\vh_t \,\vert\, \vY, \vdelta] \,\vert\, \vY] \label{eq:total_exp}\\
    &= \expec[\vW_\vdelta \vy_t \,\vert\, \vY ] \\
    &= \expec[\vW_\vdelta \,\vert\, \vY] \vy_t \label{eq:factorized_est}\\
    &= \Wmmse(\vY)\; \vy_t. \label{eq:mmse_estimation}
\end{align}
where we use the law of total expectation and~\eqref{eq:cond_mmse_estimator1} to get the final result.
We note that the observations are filtered by the MMSE estimate $\Wmmse$ of the filter $\vW_\vdelta$.
Hence, the main difficulty of the non-linear channel estimation lies with the calculation of $\Wmmse$ from the observations $\vY$.

Using Bayes' theorem to express the posterior distribution of $\vdelta$ as
\begin{equation}
p(\vdelta \vert \vY ) = \frac{p(\vY | \vdelta)p(\vdelta)}{\int p(\vY |
  \vdelta)p(\vdelta) d\vdelta}
\end{equation}
we can write the MMSE filter as
\begin{equation}\label{eq:mmseestimator}
    \Wmmse = \int p(\vdelta\vert\vY) \vW_\vdelta d\vdelta =\frac{\int p(\vY | \vdelta) \vW_\vdelta \,p(\vdelta)d\vdelta}{\int p(\vY|\vdelta)\,p(\vdelta)d\vdelta}\,.
\end{equation}
The MMSE estimation in~\eqref{eq:mmse_estimation}, \eqref{eq:mmseestimator} can be interpreted as follows.
We first calculate $\Wmmse$ as a convex combination of filters $\vW_\vdelta$ with weights $p(\vdelta\vert\vY)$ for known covariance matrices $\vC_\vdelta$ and then apply the resulting filter $\Wmmse$ to the observation.

By manipulating $p(\vY\vert\vdelta)$, we obtain the following expression for the MMSE filter, which shows that $\Wmmse$ depends on $\vY$ only through the scaled sample covariance matrix
\begin{equation}\label{eq:samplecov}
\widehat\vC = \frac{1}{\sigma^2}\sum_{t=1}^T \vy_t\vy_t\he.
\end{equation}
\begin{lemma}\label{lem:mmseestimator}
    If the noise covariance matrix is $\cz = \sigma^2 \id$, the MMSE filter $\Wmmse$ in~\eqref{eq:mmseestimator} can be calculated as
\begin{equation}\label{eq:mmseestimator_explicit}
    \Wmmse(\widehat\vC) = \frac{\int \exp\big( \tr(\vW_\vdelta\widehat\vC)
      + T\log\lvert\id -\vW_\vdelta\rvert \big)\vW_\vdelta \,p(\vdelta)d\vdelta}
  {\int\exp\big( \tr(\vW_\vdelta\widehat \vC) + T\log|\id - \vW_\vdelta| \big)\,p(\vdelta)d\vdelta}.
\end{equation}
with $\widehat \vC$ given by~\eqref{eq:samplecov} and $\vW_\vdelta$ given by~\eqref{eq:cond_mmse_estimator2}.
% with
% \begin{equation}\label{eq:mmse_filter}
% \vW_\vdelta = \vC_\vdelta(\vC_\vdelta + \sigma^2\id)^{-1}.
% \end{equation}
\end{lemma}
\begin{proof}
    See Appendix~\ref{app:lemma1}.
\end{proof}
Note that the scaled sample covariance matrix $\widehat \vC$ is a sufficient statistic to calculate the MMSE filter $\Wmmse$.
Moreover, if we define $\Hest = [\hest_1, \ldots, \hest_T]$, we see from $\Hest = \Wmmse(\widehat \vC)\; \vY$ that we use all data to construct the sample covariance matrix $\widehat\vC$ and the filter $\vW_\vdelta$ and then apply the resulting filter to each observation individually to calculate the channel estimate.
This structure is beneficial for applications in which we are only interested in the estimate of the most recent channel vector.
In such a case, we can apply an adaptive method to track the scaled sample covariance matrix.
That is, given the most recent observation $\vy$, we apply the update
\begin{align}
    \widehat\vC \assign \alpha \widehat\vC + \beta\vy\vy\he
\end{align}
with suitable $\alpha,\beta > 0$ and then calculate the channel estimate
\begin{align}
    \hest = \Wmmse(\widehat\vC)\; \vy.
\end{align}

\section{MMSE Estimation and Neural Networks}
\label{sec:mmse_learning}

For arbitrary prior distributions $p(\vdelta)$, the MMSE filter as given by Lemma~\ref{lem:mmseestimator} cannot be evaluated in closed form.
To make the filter computable, we need the following assumption.
\begin{assumption}\label{assumption:discrete_prior}
The prior $p(\vdelta)$ is discrete and uniform, i.e., we have a grid $\{\vdelta_i : i=1,\ldots,N\}$ of possible values for $\vdelta$ and
\begin{equation}
    p(\vdelta_i) = \frac{1}{N}\,, \quad \forall i=1,\ldots,N.
\end{equation}
\end{assumption}
Under this assumption, we can evaluate the MMSE estimator of $\vW_\vdelta$ as
\begin{equation}\label{eq:mmseestimator_grid}
    \Wgrid (\widehat\vC)
    = \frac{\frac{1}{N} \sum_{i=1}^N \exp\big( \tr(\vW_{\vdelta_i}\widehat\vC)
+ b_i \big)\vW_{\vdelta_i}}{ \frac{1}{N}\sum_{i=1}^N \exp\big(\tr(\vW_{\vdelta_i}\widehat\vC) + b_i \big) }
\end{equation}
where $\vW_{\vdelta_i}$ is obtained by evaluating~\eqref{eq:cond_mmse_estimator2} for $\vdelta=\vdelta_i$ and
\begin{equation}\label{eq:bdelta}
    b_i = T\log\lvert\id - \vW_{\vdelta_i}\rvert.
\end{equation}

If Assumption~\ref{assumption:discrete_prior} does not hold, e.g., if $p(\vdelta)$ describes a continuous distribution, expression~\eqref{eq:mmseestimator_grid} is only approximately true if the grid points $\vdelta_i$ are chosen as random samples from $p(\vdelta)$.
In this case, the estimator~\eqref{eq:mmseestimator_grid} is a heuristic, suboptimal estimator, which neglects that the true distribution of $\vdelta$ is continuous.
We refer to this estimator as \emph{gridded estimator} (GE).
By the law of large numbers, the approximation error vanishes as the number of samples $N$ is increased, but this also increases the complexity of the channel estimation.

We can improve the performance of the estimator for a fixed $N$ by interpreting $\vW_{\vdelta_i}$ and $\vb_i$ as variables that can be optimized instead of using the values in~\eqref{eq:cond_mmse_estimator2} and~\eqref{eq:bdelta}.
This is the idea underlying the learning-based approaches, which we describe in the following.

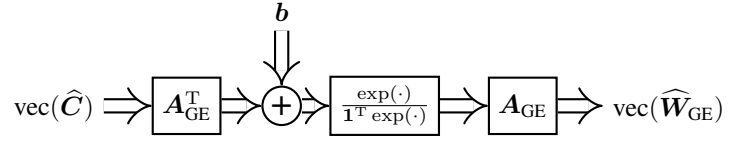
\begin{figure}[t]
    \centering
    \begin{tikzpicture}[
            blubb/.style = {regular polygon, regular polygon sides=3, fill=white, thick,scale=1.5, rotate=30},
            darrow/.style = {implies-, double distance between line centers=5pt,thick},
            dline/.style = {double distance between line centers=5pt,thick},
            sarrow/.style = {stealth-,thick},
            dspadder/.style = { circle,fill=white,draw,inner sep=0.5mm},
        ]
	\def\xdist{1.1}
	\def\ydist{1.0}
        \coordinate(y1) at (0,0);
        \draw[darrow] (y1) node[right]{$\vec(\Wgrid)$} -++ (-\xdist+0.5,0) node[convbox,left,name=V]{$\vA_\text{GE}$};
        \draw[darrow] (V) -++ (-\xdist,0) node[left,name=exp,convbox]{$\frac{\exp(\cdot)}{\ones\tp\exp(\cdot)}$};
        \draw[darrow] (exp) -++ (-\xdist,0) node[left,name=plus,dspadder]{\Large +};
        \draw[darrow] (plus) -++ (-\xdist+0.3,0) node[left,name=Vtp,convbox]{$\vA_\text{GE}\tp$};
        \draw[darrow] (Vtp) -++ (-\xdist,0) node[left]{$\vec(\widehat{\vC})$};
        \draw[darrow] (plus) -++ (0,\ydist) node[above]{$\vb$};
    \end{tikzpicture}
\caption{Block diagram of the gridded estimator $\Wgrid$.}
\label{fig:neuron}
\end{figure}

\begin{figure}[t]
    \centering
    \begin{tikzpicture}[
            blubb/.style = {regular polygon, regular polygon sides=3, fill=white, thick,scale=1.5, rotate=30},
            darrow/.style = {implies-, double distance between line centers=5pt,thick},
            dline/.style = {double distance between line centers=5pt,thick},
            sarrow/.style = {stealth-,thick},
            dspadder/.style = { circle,fill=white,draw,inner sep=0.5mm},
        ]
	\def\xdist{1.1}
	\def\ydist{1.0}
        \coordinate(y1) at (0,0);
        \draw[darrow] (y1) node[right]{$\vw$} -++ (-\xdist+0.5,0) node[left,name=plus2,dspadder]{\Large +};
        \draw[darrow] (plus2)  -++ (-\xdist+0.25,0) node[convbox,left,name=V]{$\vA^{(2)}$};
        \draw[darrow] (V) -++ (-\xdist,0) node[left,name=exp,convbox]{$\phi(x)$};
        \draw[darrow] (exp) -++ (-\xdist-0.1,0) node[left,name=plus,dspadder]{\Large +};
        \draw[darrow] (plus) -++ (-\xdist+0.25,0) node[left,name=Vtp,convbox]{$\vA^{(1)}$};
        \draw[darrow] (Vtp) -++ (-\xdist,0) node[left]{$\vx$};

        \draw[darrow] (plus) -++ (0,\ydist) node[above]{$\vb^{(1)}$};
        \draw[darrow] (plus2) -++ (0,\ydist) node[above]{$\vb^{(2)}$};
    \end{tikzpicture}
    \caption{Neural network with two layers and activation function $\phi(x)$.}
\label{fig:nn}
\end{figure}
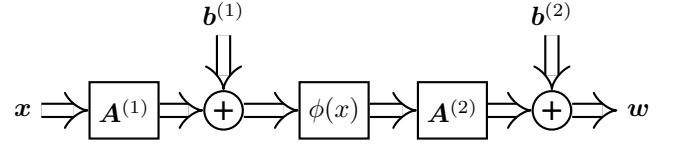

Let us first analyze the structure of the gridded estimator.
If we consider the vectorization of $\Wgrid(\widehat\vC)$, i.e.,
\begin{align}
    \vec(\Wgrid(\widehat\vC)) = \vA_\text{GE} \frac{\exp(\vA_\text{GE}\tp \vec(\widehat\vC) + \vb)}{\ones\tp \exp(\vA_\text{GE}\tp \vec(\widehat\vC) + \vb)}
\end{align}
where
%\begin{align}
$    \vA_\text{GE} = [\vec(\vW_{\vdelta_1}), \ldots, \vec(\vW_{\vdelta_N})] \in \C^{M^2 \times N}$
%\end{align}
and 
%\begin{align}
$    \vb = [b_1, \ldots, b_N]$,
%\end{align}
we see that the function~\eqref{eq:mmseestimator_grid} can be visualized as the block diagram shown in Fig.~\ref{fig:neuron}.
A slightly more general structure is depicted in Fig.~\ref{fig:nn}, which is readily identified as a common structure of a feed-forward neural network (NN) with two \emph{linear layers}, which are connected by a nonlinear \emph{activation function}.
The gridded estimator $\Wgrid$ is a special case of the neural network in Fig.~\ref{fig:nn}, which uses the \emph{softmax} function 
\begin{align}
    \phi(\vx) = \frac{\exp(\vx)}{\ones\tp\exp(\vx)}
\end{align}
as activation function and the specific choices $\vA^{(1)} = \vA_\text{GE}\tp$, $\vA^{(2)} = \vA_\text{GE}$, $\vb^{(1)} = \vb$ and $\vb^{(2)} = \zeros$ for the variables.

To formulate the learning problem mathematically, we define the set of all functions that can be represented by the NN in Fig.~\ref{fig:nn} as
\begin{align}
    &\setWgrid =\\
 \bigg\{ &\vf(\cdot) : \C^{M^2} \mapsto \C^{M^2},\;
                        \vf(\vx) = \vA^{(2)} \phi(\vA^{(1)}\vx + \vb^{(1)}) + \vb^{(2)}, \notag \\
                        &\vA^{(1)} \in \C^{N\times {M^2}}, \vA^{(2)} \in \C^{{M^2}\times N},
                        \vb^{(1)} \in \C^N,            \vb^{(2)} \in \C^{M^2}\notag
    \bigg\}.
\end{align}
The MSE of a given estimator $\widehat\vW(\cdot)$, which takes the scaled covariance matrix $\widehat\vC$ as input, is given by
\begin{align}
    \eps(\widehat\vW(\cdot)) = \expec[\|\vH-\widehat\vW(\widehat\vC)\;\vY\|^2_F].
\end{align}
The optimal neural network, i.e., the NN-MMSE estimator, is given as the function in the set $\setWgrid$ that minimizes the MSE,
\begin{equation}\label{eq:wgrid_opt}
    \vec(\Wgridstar(\cdot)) = \argmin_{\vec(\widehat\vW(\cdot))\in\setWgrid} \eps(\widehat\vW(\cdot)). 
\end{equation}
Since we assume that the dimension $N$ and the activation function $\phi(\cdot)$ are fixed, the variational problem in~\eqref{eq:wgrid_opt} is simply an optimization over the variables $\vA^{(\ell)}$ and $\vb^{(\ell)}$, $\ell=1,2$.

If we choose the softmax function as activation function, and if Assumption~\ref{assumption:discrete_prior} is fulfilled, we have 
\begin{align}
    \eps(\Wgrid(\cdot)) = \eps(\Wgridstar(\cdot)) = \eps(\Wmmse(\cdot))
\end{align}
since, in this case, the gridded estimator is the MMSE estimator, $\Wgrid(\cdot)=\Wmmse(\cdot)$, and because $\vec(\Wgrid(\cdot))\in\setWgrid$.
In general, we have the relation
\begin{align}
    \eps(\Wgrid(\cdot)) \geq \eps(\Wgridstar(\cdot)) \geq \eps(\Wmmse(\cdot)).
\end{align}

The optimization problem~\eqref{eq:wgrid_opt} is a typical learning problem for a neural network with a slightly unusual cost function.
Due to the expectation in the objective function, we have to revert to stochastic gradient methods to find (local) optima for the variables of the neural network.
Unlike the \emph{gridded estimator}~\eqref{eq:mmseestimator_grid}, which relies on analytic expressions for the covariance matrices $\vC_\vdelta$, the \emph{neural network estimator} merely needs a large data set $\{(\vH_1,\vY_1), (\vH_2,\vY_2), \ldots \}$ of channel realizations and corresponding observations to optimize the variables.
In fact, we could also take samples of channel vectors and observations from a measurement campaign to learn the NN-MMSE estimator for the ``true'' channel model.
This requires that the SNR during the measurement campaign is significantly larger than the SNR in operation.
If, as assumed, the noise covariance matrix is known, the observations can then be generated by adding noise to the channel measurements.

The basic structure of the NN-MMSE estimator is depicted in Fig.~\ref{fig:estimator}.
The learning of the optimal variables $\vA^{(\ell)}$ and $\vb^{(\ell)}$ is performed off-line and needs to be done only once.
During operation, the channel estimates are obtained by first forming the scaled sample covariance matrix $\widehat\vC$, which is then fed into the neural network $\Wgridstar(\cdot)$.
Finally, the output $\Wgridstar(\widehat\vC)$ of the neural network is applied as a linear filter to the observations $\vY$ to get the channel estimates $\Hest$.

\begin{figure}[t]
    \centering
    \begin{tikzpicture}[
        lintrans/.style = {fill=white,draw,isosceles triangle, isosceles triangle apex angle=58, minimum width=0.8cm},
        darrow/.style = {implies-, double distance between line centers=5pt,thick},
        dline/.style = {double distance between line centers=5pt,thick},
        sarrow/.style = {stealth-,thick},
        dspadder/.style = { circle,fill=white,draw,inner sep=0.5mm},
        convbox/.style = {fill=white,draw,minimum width=1cm,minimum height=0.75cm},
        ]
	\def\xdist{1.8}
	\def\ydist{2.0}
        \coordinate(y1) at (0,0);
        \draw[darrow] (y1) node[right]{$\Hest$} -- ++(-\xdist,0) node[convbox,name=V]{$\widehat{\bm W}$};
        \draw[darrow] (V) -- ++(-\xdist-\xdist,0) coordinate (intersection) -- ++(-\xdist,0) node[left]{$\vY$};
        \draw[darrow] ($(V.south) + (0.4,-0.5)$) -- ($(V.north) + (-0.4,0.4)$) |- ($(V) + (-\xdist+0.2,\ydist)$) node[midway,right,shift={(0.1,-0.8)}]{$\widehat\vW=\Wgridstar(\widehat\vC)$} node[convbox,name=structure]{$\Wgridstar(\cdot)$};
        \draw[darrow] (structure) -| ($(intersection) + (0,1)$) node[midway,left,shift={(-0.1,-0.1)}]{$\widehat\vC$} node[convbox,name=cov]{sample cov.};
        \draw[darrow] (cov) -- ($(intersection) + (0,2.25pt)$);
        \draw[thick] (structure) node[convbox,name=learning,above,shift={(-2.6,0.5)}]{\small off-line learning};
        \draw[sarrow] ($(structure.south) + (0.3,-0.3)$) -- ($(structure.north) + (-0.3,0.5)$) -| (learning.east) node[midway,above,shift={(0.8,0)}]{$\vA^{(\ell)},\vb^{(\ell)}$};
        \draw[darrow] (structure) node[convbox]{$\Wgridstar(\cdot)$};
        \draw[darrow] (V) node[convbox]{$\widehat{\bm W}$};
    \end{tikzpicture}
    \caption{Channel estimator with embedded neural network.}
\label{fig:estimator}
\end{figure}
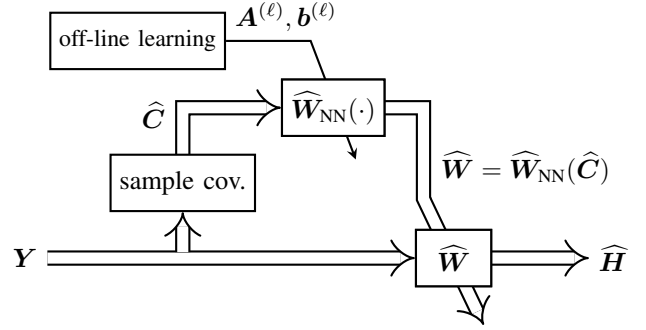

With proper initialization and sufficient quality of the training data, the neural network estimator is guaranteed to outperform the gridded estimator, which has the same computational complexity.
However, there are two problems with this learning approach, which we address in the following sections.
First, finding the optimal neural network $\Wgridstar$ is too difficult, because the number of variables is huge and the optimization problem is not convex.
Second, even if the optimal variables were known, the computation of the channel estimate $\hest_t = \Wgridstar(\widehat\vC)\,\vy_t$ is too complex:
Evaluating the output of the neural network $\Wgridstar(\widehat\vC)$ needs $\order(M^2N)$ floating point operations due to the matrix-vector products (cf.~Fig.~\ref{fig:nn}).
For example, if the grid size $N$ needs to scale linearly with the number of antennas $M$ to obtain accurate estimates, the computational complexity scales as $\order(M^3)$, which is too high for practical applications.

\section{Low-complexity MMSE Estimation}

With Assumption~\ref{assumption:discrete_prior} the gridded estimator $\Wgrid$ in~\eqref{eq:mmseestimator_grid} is the MMSE estimator.
In the following, we introduce additional assumptions, which help to simplify $\Wgrid$.
With these assumptions, we get a fast channel estimator, i.e., one with a computational complexity of only $\order(M\log M)$.
Just as for the gridded estimator, the fast estimator is no longer the MMSE estimator if the assumptions are violated.
However, in analogy to Sec.~\ref{sec:mmse_learning}, the structure of this fast estimator motivates the convolutional neural network (CNN) estimator presented in Sec.~\ref{sec:neural_network}.

Our approach to reduce the complexity of $\Wgrid$ can be broken down into two steps.
First, we exploit common structure of the covariance matrices, which occur for commonly used array geometries.
In a second step, we use an approximated shift-invariance structure, which is present in a certain channel model with only a single path of propagation.
With those two steps, we reduce the computational complexity from $\order(M^2N)$ to $\order(M\log M)$.

\subsection{A Structured MMSE Estimator}
\label{sec:structured_mmse}

In the first step, we replace the filters $\vW_{\vdelta_i}$ in~\eqref{eq:mmseestimator_grid} with structured matrices that use only $\order(M)$ variables.
Specifically, we make the following assumption. 
\begin{assumption}\label{assumption:diag_structure}
    The filters $\vW_{\vdelta_i}$ can be decomposed as
    \begin{equation}\label{eq:diag_structure}
        \vW_{\vdelta_i} = \vQ\he \diag( \vw_i ) \vQ
    \end{equation}
    with a common matrix $\vQ \in \C^{K\times M}$ and vectors
    $\vw_i \in \R^K$ where $\order(K)=\order(M)$.
\end{assumption}

Note that the requirement $\order(K) = \order(M)$ ensures the desired dimensionality reduction and $\vw_i\in\R^K$ leads to self-adjoint filters.
Combining Assumptions~\ref{assumption:discrete_prior} and~\ref{assumption:diag_structure}, we get the following result.
\begin{thm}\label{thm}
    Given Assumptions~\ref{assumption:discrete_prior} and~\ref{assumption:diag_structure}, the MMSE estimator of $\vW_\vdelta$ simplifies to
    \begin{align}\label{eq:Wstruct}
        \Wstruct(\widehat\vC) = \vQ\he \diag(\wstruct(\hat\vc)) \vQ
    \end{align}
    where
    \begin{equation}\label{eq:sample_spectrum_general}
        \hat \vc = \frac{1}{\sigma^2} \sum_{t=1}^T \abs{\vQ\vy_t}^2.
    \end{equation}
    Moreover, the element-wise filter $\wstruct$ is given by
    \begin{equation}\label{eq:diag_opt_filt}
        \wstruct(\hat \vc) = \vA_\text{SE} \frac{\exp(\vA_\text{SE}\tp \hat \vc + \vb)}{\ones\tp \exp(\vA_\text{SE}\tp \hat \vc + \vb)}
    \end{equation}
    where the matrix
    \begin{equation}\label{eq:amat}
        \vA_\text{SE} = [\vw_i, \ldots, \vw_N ] \in \R^{K\times{N}}
    \end{equation}
    contains the element-wise MMSE filters~\eqref{eq:diag_structure}, and the entries of the vector
    \begin{equation}
      \vb = [b_1, \ldots, b_N ]\tp \in \R^N
    \end{equation}
    are given by~\eqref{eq:bdelta}.
\end{thm}

\begin{proof}
If we replace the filters $\vW_{\vdelta_i}$ in~\eqref{eq:mmseestimator_grid} with the parametrization in~\eqref{eq:diag_structure}, we can simplify the trace expressions according to
\begin{align}
    \tr(\vW_{\vdelta_i} \widehat \vC)
&= \tr(\vQ\he\diag(\vw_i) \vQ \widehat\vC)\\
&= \tr\Big(\diag(\vw_i) \frac{1}{\sigma^2}\sum_{t=1}^T\vQ \vy_t^{} \vy_t\he\vQ\he\Big)\\
&= \vw_i\tp \hat \vc
\end{align}
as $\hat \vc$ contains the diagonal elements of the matrix 
\begin{align}
    \frac{1}{\sigma^2}\sum_{t=1}^T\vQ \vy_t^{} \vy_t\he\vQ\he.
\end{align}
Consequently, the gridded estimator in~\eqref{eq:mmseestimator_grid} simplifies to
\begin{align}\label{eq:mmse_filt_circ}
    \Wstruct(\hat\vc) &= \frac{\sum_{i=1}^N\exp( \vw_i\tp\hat\vc + b_i ) \vQ\he \diag(\vw_i) \vQ}{\sum_{i=1}^N\exp(\vw_i\tp\hat \vc + b_i )} \notag \\
                      &= \vQ\he \diag\left(\frac{\sum_{i=1}^N\exp( \vw_i\tp\hat\vc + b_i ) \vw_i }{\sum_{i=1}^N\exp(\vw_i\tp\hat \vc + b_i )}\right) \vQ .
\end{align}
With the definitions of $\vA_\text{SE}$ and $\vb$, we can write~\eqref{eq:mmse_filt_circ} as~\eqref{eq:Wstruct}.
\end{proof}

If Assumptions~\ref{assumption:discrete_prior} and~\ref{assumption:diag_structure} hold, the MMSE estimates of the channel vectors using the \emph{structured estimator} (SE) can be calculated as
\begin{equation}
    \hest_t = \vQ\he \diag(\wstruct(\hat\vc)) \vQ \vy_t
\end{equation}
i.e., $\Wmmse(\widehat\vC) = \vQ\he\diag(\wstruct(\hat\vc)\vQ$.

Given $\wstruct(\hat\vc)$, the complexity of the estimator depends only on the number of operations required to calculate matrix-vector products with $\vQ$ and $\vQ\he$.
To achieve the desired complexity $\order(M\log M)$, the matrix $\vQ$ must have some special structure that enables fast computations.
If this is the case, the complexity of the \emph{structured estimator} is dominated by the calculation of $\wstruct(\hat\vc)$, which is $\order(NK)$.
In Sec.~\ref{sec:fast_mmse}, we show how the complexity of the calculation of $\wstruct(\hat\vc)$ can be reduced further.

\textbf{Examples.} 
For a uniform linear array (ULA), the channel covariance matrices, which have Toeplitz structure, are asymptotically equivalent to corresponding circulant matrices~(cf.~\cite{gray_2006_toeplitz}, Appendix~\ref{app:ula}).
Since all circulant matrices have the columns of the discrete Fourier transform (DFT) matrix $\vF$ as eigenvectors, we have the asymptotic equivalence 
\begin{equation}
    \vC_\vdelta \asymp \vF\he \diag(\vc_\vdelta) \vF \;\;\forall \vdelta
\end{equation}
where $\vc_\vdelta$ contains the diagonal elements of $\vF \vC_\vdelta \vF\he$.
As a consequence, we have a corresponding asymptotic equivalence 
\begin{equation}
    \vW_\vdelta \asymp \vF\he \diag(\vw_\vdelta) \vF \;\;\forall \vdelta
\end{equation}
where $\vw_\vdelta$ contains the diagonal elements of $\vF \vW_\vdelta \vF\he$.
For a large-scale system, this is a very good approximation~\cite{epstein_how_2005}.
We call the structured estimator that uses Assumption~\ref{assumption:diag_structure} with $\vQ = \vF$ the \emph{circulant estimator}.

To reduce the approximation error for finite numbers of antennas, we can use a more general factorization with $\vQ=\vF_2$, where $\vF_2 \in \C^{2M \times M}$ contains the first $M$ columns of a $2M\times 2M$ DFT matrix.
The class of matrices that can be expressed as
\begin{equation}
    \vW_\vdelta = \vF_2\he \diag(\vw_\vdelta) \vF_2
\end{equation}
are exactly the Toeplitz matrices~\cite{epstein_how_2005}.
Note that the filters $\vW_\vdelta$ do not actually have Toeplitz structure, even if the channel covariance matrices are Toeplitz matrices.
The Toeplitz assumption only holds in the limit for large numbers of antennas due to the arguments given above or for low SNR when the noise covariance matrix dominates the inverse in~\eqref{eq:cond_mmse_estimator2}.
Nevertheless, the Toeplitz structure is more general than the circulant structure and, thus, yields a smaller approximation error.
The estimator that uses Assumption~\ref{assumption:diag_structure} with $\vQ = \vF_2$ is denoted as the \emph{Toeplitz estimator}. 

An analogous result can be derived for uniform rectangular arrays (cf.~Appendix~\ref{app:ura}).
In this case, the transformation $\vQ$ is the Kronecker product of two DFT matrices, whose dimensions correspond to the number of antennas in both directions of the array.

A third example with a decomposition as in Assumption~\ref{assumption:diag_structure} is a setup with distributed antennas~\cite{ngo_cell-free_2015,ngo_cell-free_2017}.
For distributed antennas, the covariance matrices are typically
modelled as diagonal matrices and, thus, the filters $\vW_\vdelta$ are diagonal as well.
That is, for distributed antennas we simply have $\vQ=\id$.

\subsection{A Fast MMSE Estimator}
\label{sec:fast_mmse}

The main complexity in the evaluation of $\wstruct(\cdot)$ stems from the matrix-vector products in~\eqref{eq:diag_opt_filt}.
The complexity can be reduced by using only matrices $\vA_\text{SE}$ that allow for fast matrix-vector products.
One possible choice are the circulant matrices.

In fact, circulant matrices naturally arise in the structured estimator for a single-path channel model with a single parameter $\delta$ for the angle of arrival.
In this model, the power spectrum is shift-invariant, i.e., $g(\theta;\delta) = g(\theta - \delta)$.
As a result, for $N=K$, the samples $\vw_i$ of the structured estimator $\wstruct$ in~\eqref{eq:diag_opt_filt} are approximately shift invariant, i.e., their entries satisfy $[\vw_i]_j = [\vw_{i+n}]_{j+n}$ (the sums are modulo $M$) and the following assumption is satisfied (more details are given in Appendix~\ref{app:shift_invariance}).
\begin{assumption}\label{assumption:shift_invariance}
    The matrix $\vA_\text{SE}$ in~\eqref{eq:amat} is circulant and given by
\begin{equation}\label{eq:fast_estimator_trafo}
    \vA_\text{SE} = \vF\he \diag(\vF \vw_0) \vF 
\end{equation}
for some $\vw_0 \in \R^K$, where $\vF$ is the $K$-dimensional DFT matrix.
\end{assumption}

Note that Assumption~\ref{assumption:shift_invariance} is, in principle, independent of Assumption~\ref{assumption:diag_structure}.
We see from the examples that the structure of $\vW_\vdelta$ and, thus, the choice for $\vQ$ is motivated by the array geometry, while the assumption that $\vA_\text{SE}$ is circulant is motivated by the physical channel model.
The example in Appendix~\ref{app:shift_invariance}, which is based on the ULA geometry and the 3GPP channel model, suggests a circulant structure for both the filters $\vW_\vdelta$ and the matrix $\vA_\text{SE}$ in $\wstruct(\cdot)$.

However, we could think of other system setups, where the structure of $\vW_\vdelta$ in Assumption~\ref{assumption:diag_structure} is different than the structure of $\vA_\text{SE}$ in Assumption~\ref{assumption:shift_invariance}.
As an illustration, consider a toy example where we have an array of antennas along a long corridor, say in an airplane.
Then we could reasonably assume diagonal covariance matrices, i.e., $\vQ=\id$, but at the same time we have a shift-invariance for different positions of the users in the corridor, i.e., Assumption~\ref{assumption:shift_invariance} also holds.

Given the relationship  between circulant matrices and circular convolution, we can write
\begin{equation}\label{eq:circulant_parametrization}
    \vA_\text{SE} \vx = \vF\he \diag(\vF \va) \vF\vx = \va \ast \vx
\end{equation}
with $\va \in \R^K$.
Because of the FFT, the computational complexity of evaluating $\wstruct(\hat\vc)$ reduces to $\order(M\log M)$ if $\order(K)=\order(M)$.
That is, we get a \emph{fast estimator} (FE)
\begin{equation}\label{eq:fast_opt_filt}
    \wshift(\hat\vc) = \vw_0 \ast \text{softmax}(\tilde\vw_0 \ast \hat\vc + \vb)
\end{equation}
by incorporating the constraint~\eqref{eq:fast_estimator_trafo} into $\wstruct$. 
The vector $\tilde\vw_0$ contains the entries of $\vw_0$ in reversed order.

\section{Low-complexity Neural Network}
\label{sec:neural_network}

For most channel models, Assumptions~\ref{assumption:discrete_prior}, \ref{assumption:diag_structure}, and~\ref{assumption:shift_invariance} only hold approximately or even not at all.
That is, the estimator $\wshift$ in~\eqref{eq:fast_opt_filt} does not yield the MMSE estimator in most practical scenarios.
Nevertheless, it is still worthwhile to consider an estimator with similar structure:
Calculating a channel estimate with $\Wgrid$ costs $\order(M^2N)$ FLOPS, while using $\wshift$ only requires $\order(M\log M)$ operations.
As we discuss in the following, another advantage is that the number of variables that have to be learned reduces from $\order(M^2N)$ to $\order(M)$, since we no longer have full matrices, but circular convolutions.

In Sec.~\ref{sec:mmse_learning} we discussed how learning can be used to compensate for the approximation error that results from a finite grid size $N$, i.e., a violation of Assumption~\ref{assumption:discrete_prior}.
Analogously, we can learn the variables of a convolutional neural network inspired by $\wshift$ to compensate for violations of Assumptions~\ref{assumption:diag_structure} and~\ref{assumption:shift_invariance}.
To this end, we define the set of CNNs 
\begin{align}\label{eq:setwshift}
\setWshift = \Bigg\{
    &\vx \mapsto \va^{(2)} \ast \phi\Big(\va^{(1)}\ast \vx + \vb^{(1)}\Big) + \vb^{(2)}\,,\notag\\
    &\va^{(\ell)}\in\set \R^K, \vb^{(\ell)}\in\R^K, \ell=1,2
\Bigg\}
\end{align}
and the optimal \emph{CNN estimator} is the one using
\begin{equation}\label{eq:wshift_opt}
    \wshiftstar(\cdot) = \argmin_{\hat{\vw}(\cdot) \in\setWshift} \eps(\vQ\he \hat\vw(\cdot) \vQ).
\end{equation}
Again, we assume that the activation function $\phi(\cdot)$ is fixed.
Thus, the optimization is only with respect to the convolution kernels $\va^{(\ell)}$ and the bias vectors $\vb^{(\ell)}$.

Analogously to Sec.~\ref{sec:mmse_learning}, if we choose the softmax function as activation function, we have $\wshift \in \setWshift$.
Consequently, if Assumptions~1--3 are fulfilled we get
\begin{align}
    \eps(\vQ\he\wshift(\cdot)\vQ) = \eps(\vQ\he\wshiftstar(\cdot)\vQ) = \eps(\Wmmse(\cdot)).
\end{align}
In general, we have
\begin{align}
    \eps(\vQ\he\wshift(\cdot)\vQ) \geq \eps(\vQ\he\wshiftstar(\cdot)\vQ) \geq \eps(\Wmmse(\cdot)).
\end{align}

The stochastic-gradient method that learns the CNN is described in detail in Alg.~\ref{alg:learnedfastmmse}.
We want to stress again that the learning procedure is performed off-line and does not add to the complexity of the channel estimation.
During operation, the channel estimation is performed by evaluating $\wshiftstar(\hat\vc)$ and the transformations involving the $\vQ$ matrix for given observations.
If the variables are learned from simulated samples according to the 3GPP or any other channel model, this algorithm suffers from the same model-reality mismatch as does any other model-based algorithm.
The fact that the proposed algorithm can also be trained on true channel realizations puts it into a significant advantage over other non-learning-based algorithms, which have to rely on models only. 

\begin{algorithm}[t]
\begin{algorithmic}[1]
%    A simplified channel model, e.g., the one from Sec.~\ref{sec:single_path} can be used for the initialization.
    \State Initialize variables $\va^{(\ell)}$ and $\vb^{(\ell)}$ randomly
    \State Generate/select a mini-batch of $S$ channel vectors $\vH_s$ and corresponding observations $\vY_s$ (and $\hat \vc_s$) for $s = 1,$ \dots, $S$ 
    \State Calculate the stochastic gradient
        \[
            \vg = \frac{1}{S}\sum_{s=1}^S \frac{\partial}{\partial [\va^{(\ell)};\vb^{(\ell)}]} \norm{ \vH_s-\vQ\he \diag(\hat\vw(\hat\vc_s)) \vQ\vY_s}_F^2
        \]
        with $\hat\vw(\vx)$ as stated in~\eqref{eq:fast_opt_filt}
    \State Update variables with a gradient algorithm (e.g.,~\cite{diederik_adam_2014})
    \State Repeat steps 1--3 until a convergence criterion is satisfied
\end{algorithmic}
\caption{Learned fast MMSE filter}\label{alg:learnedfastmmse}
\end{algorithm}

In the simulations, we compare two variants of the CNN estimator.
First, we use the softmax activation function $\phi=\frac{\exp(\cdot)}{\ones\tp\exp(\cdot)}$.
The resulting softmax CNN estimator is a direct improvement over the fast estimator with $\wshift$, which was derived under Assumptions 1--3.
In the second variant, we use a rectified linear unit (ReLU) $\phi(x) = [x]_+$ as activation function, since ReLUs were found to be easier to train than other activation functions~\cite{Glorot11} (and they are also easier to evaluate than the softmax function).

\subsection{Hierarchical Learning}
\label{sec:learning}

\begin{algorithm}[t]
    \begin{algorithmic}[1]
    \State Choose upsampling factor $\beta > 1$ and number of stages $n$
    \State Set $M_0=\lceil M/\beta^n\rceil$, $K_0=\lceil K/\beta^n\rceil$
    \State Learn optimal $\va_0^{(\ell)}, \vb_0^{(\ell)} \in \R^{K_0}$ using Alg.~\ref{alg:learnedfastmmse} assuming $M_0$ antennas with random initializations
    \For{$i$ from $1$ to $n$}
        \State Set $M_i = \lceil M/\beta^{n-i}\rceil$ and $K_i = \lceil K/\beta^{n-i} \rceil$
        \State Interpolate $\va_i^{(\ell)},\vb^{(\ell)}_i \in \R^{K_i}$ from $\va_{i-1}^{(\ell)},\vb^{(\ell)}_{i-1} \in \R^{K_{i-1}}$
        \State Normalize $\va_i^{(\ell)}$ by dividing by $\beta$
        \State Learn optimal $\va_i^{(\ell)}, \vb_i^{(\ell)}$ using Alg.~\ref{alg:learnedfastmmse} assuming $M_i$
\hspace*{\algorithmicindent}antennas and using $\va_i^{(\ell)}, \vb_i^{(\ell)}$ as initializations
    \EndFor
\end{algorithmic}
\caption{Hierarchical Training}\label{alg:hierarchical_learning}
\end{algorithm}

\makeatletter
\pgfplotsset{
    boxplot/hide outliers/.code={
        \def\pgfplotsplothandlerboxplot@outlier{}%
    }
}
\makeatother
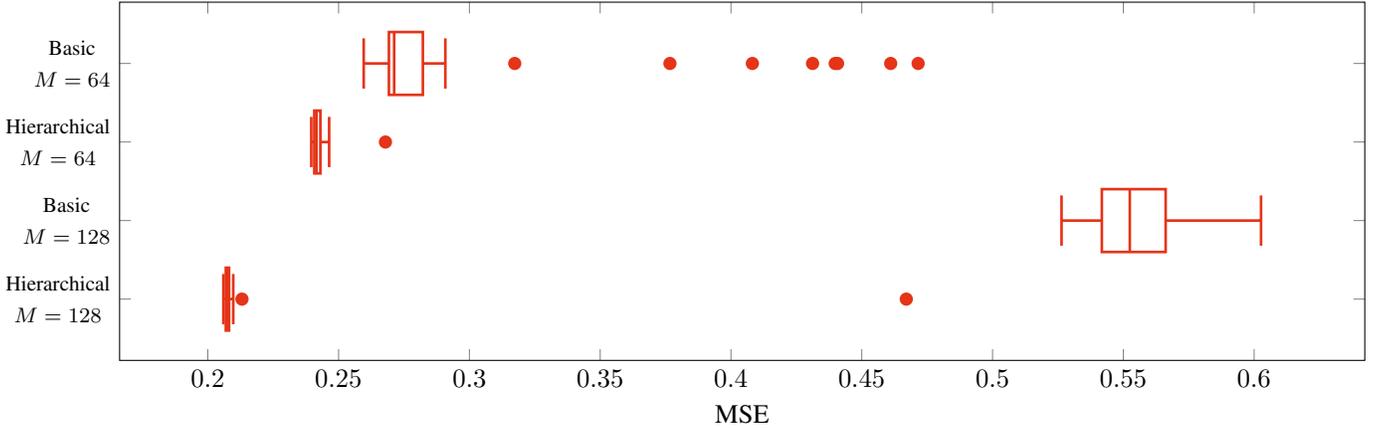
\begin{figure*}[t]
    \centering
    \begin{tikzpicture}
\begin{axis}[
    height=0.35\textwidth,
    width=\textwidth,
    xlabel= MSE,
    yticklabel style={align=center},
    ytick = {1,2,3,4},
    yticklabels = {
        {\footnotesize Hierarchical\\\footnotesize $M=128$},
        {\footnotesize Basic\\\footnotesize $M=128$},
        {\footnotesize Hierarchical\\\footnotesize $M=64$},
        {\footnotesize Basic\\\footnotesize $M=64$},
    },
]

\addplot[
    boxplot,
    line width = 1.0pt,
    color = TUMBeamerRed,
    mark = *,
    ] table [ignore chars=",y index=0,col sep=comma] {figures/figure4-CircReLUHier-nAnt-128.csv};
\addplot[
    boxplot,
    line width = 1.0pt,
    color = TUMBeamerRed,
    mark = *,
    ] table [ignore chars=",y index=0,col sep=comma] {figures/figure4-CircReLU-nAnt-128.csv};
\addplot[
    boxplot,
    line width = 1.0pt,
    color = TUMBeamerRed,
    mark = *,
    ] table [ignore chars=",y index=0,col sep=comma] {figures/figure4-CircReLUHier-nAnt-64.csv};
\addplot[
    boxplot,
    line width = 1.0pt,
    color = TUMBeamerRed,
    mark = *,
    ] table [ignore chars=",y index=0,col sep=comma] {figures/figure4-CircReLU-nAnt-64.csv};

\end{axis}
\end{tikzpicture}
    \caption{
        Box plot with outliers (marked as dots) of the MSE after learning for $10\,000$ iterations for hierarchical and non-hierarchical learning. 
        We show results for $M=64$ and $M=128$ antennas for $50$ data points per plot and with the DFT matrix $\vQ=\vF$ for the transformation.
        Scenario with three propagation paths, $\sigma^2 = 1$, $T=1$.
    }
    \label{fig:hierarchical_learning}
\end{figure*}

Local optima are a major issue when learning the neural networks, i.e., when calculating a solution of the nonlinear optimization problem~\eqref{eq:wshift_opt}.
During our experiments, we observed that, especially for a large number of antennas, the learning often gets stuck in local optima.
To deal with this problem, we devise a hierarchical learning procedure that starts the learning with a small number of antennas and then increases the number of antennas step by step.

For the single-path channel model, which motivates the circulant structure of the matrices $\vA^{(\ell)}$, the convolution kernel $\vw_0$ contains samples of the continuous function $w(u;0)$, i.e., $[\vw_0]_k = w(2\pi(k-1)/K; 0)$ (cf.~Appendix~\ref{app:shift_invariance}).
If we assume that $w(u;0)$ is a smooth function, we can quite accurately calculate the generating vector $\vw_0$ for a system with $M$ antennas from the corresponding vector of a system with less antennas by commonly used interpolation methods.

This observation inspires the following heuristic for initializing the variables $\va^{(\ell)}$ and $\vb^{(\ell)}$ of a K-dimensional CNN.
We first learn the variables of a smaller CNN, e.g., we choose a CNN with dimension $K/2$.
We use the resulting variables to initialize every second entry of the vectors $\va^{(\ell)}$ and $\vb^{(\ell)}$.
The remaining entries can be obtained by numerical interpolation.

For the filter $\wshiftstar(\cdot)$, it is desirable to have outputs of similar magnitude, irrespective of the dimension $K$.
By doubling the number of entries of the convolution kernels via interpolation, we approximately double the largest absolute value of $\va^{(\ell)} \ast \vx$.
To remedy this issue, we normalize the kernels of the convolution after the interpolation such that we get approximately similar values at the outputs of each layer.
This heuristic leads to the hierarchical learning described in Alg.~\ref{alg:hierarchical_learning}.

The hierarchical learning significantly improves convergence speed and reduces the computational complexity per iteration due to the reduced number of antennas in many learning steps.
In fact, for a large number of antennas the hierarchical learning is essential to obtain good performance.
In Fig.~\ref{fig:hierarchical_learning}, we show a standard box plot~\cite{frigge_implementations_1989} of the MSE of the estimators obtained by applying the hierarchical and the standard learning procedure.
Each data point used to generate the box plot corresponds to a randomly initialized estimator and one run of Alg.~\ref{alg:hierarchical_learning} with $\beta=2$ and $n=3$ for the hierarchical learning and $n=0$ for the non-hierarchical learning (and the same total number of iterations).
The box plot depicts a summary of the resulting distribution, showing the median and the quartiles in a box and outliers outside of the whiskers as additional dots.
The whiskers are at the position of the lowest and highest data point within a distance from the box of 1.5 times the box size.
As we can see, without the hierarchical learning, the learning procedure gets stuck in local optima.
With the hierarchical approach, we are less likely to be caught in local optima during the learning process.

\section{Related Work}

In this section, we give a short summary of two alternative channel estimation methods with $\order(M\log M)$ complexity.
These methods will serve as a benchmark in the numerical evaluation of our novel algorithms.

\subsection{ML Covariance Matrix Estimation}\label{sec:ml}

The common approach to approximate MMSE channel estimation for unknown covariance matrices is to use a maximum likelihood (ML) estimate of the channel covariance matrix.
That is, we first find an ML estimate of the channel covariance matrix $\vC_\vdelta^\text{ML}$ based on the observations $\vY$ and then, assuming the estimate is exact, calculate the MMSE estimates of the channel vectors as in~\eqref{eq:cond_mmse_estimator1}, \eqref{eq:cond_mmse_estimator2}.
The disadvantage of ML estimation is that a general prior $p(\vdelta)$ cannot be incorporated.

The likelihood function for the channel covariance matrix given the noise covariance matrix is
\begin{multline}\label{eq:likelihood}
    L(\vC_\vdelta\vert \vY) \\= \exp\Big( -\sum_{t=1}^T \vy_t\he (\vC_\vdelta + \cz)\inv \vy_t - T \log\abs{\vC_\vdelta + \cz}\Big) \frac{1}{\pi^{MT}}
\end{multline}
and the ML problem reads as
\begin{equation}
    \vC_\vdelta^\text{ML} = \argmax_{\vC_\vdelta \in \set M} L(\vC_\vdelta \vert \vY)
\end{equation}
where $\set M$ is the set of admissible covariance matrices, which has to be included in the set of positive semi-definite matrices $\set S_0^+$, i.e., $\set M \subset \set S_0^+$.

If $\set M = \set S_0^+$, the ML estimate is given in terms of the sample covariance matrix
\begin{equation}
    \vS = \frac{1}{T} \sum_{t=1}^T \vy_t\vy_t\he
\end{equation}
as
\begin{equation}
    \vC_\vdelta^\text{ML} = \cz^{1/2} P_{\set S_0^+}\left( \cz^{-1/2} \vS \cz^{-1/2} -\id \right) \cz^{1/2}
\end{equation}
where we use the projection $P_{\set S_0^+}(\cdot)$ onto the cone of positive semi-definite matrices~\cite{neumann_low-complexity_2015}.
The projection $P_{\set S_0^+}(\vX)$ of a hermitian matrix $\vX$ replaces all negative eigenvalues of $\vX$ with zeros.
For $\cz = \sigma^2 \id$, the estimate simplifies to
\begin{equation}
    \vC_\vdelta^\text{ML} = P_{\set S_0^+}\left( \widehat \vC - \sigma^2\id \right).
\end{equation}

\subsubsection*{Low-Complexity ML Estimation}

If we have a ULA at the base station, we know that the covariance matrix has to be a Toeplitz matrix.
Thus, we should choose $\set M = \set T_0^+$ as the set of positive semi-definite Toeplitz matrices.
In this case, the ML estimate can no longer be given in closed form and iterative methods have to be used~\cite{burg_estimation_1982,anderson_asymptotically_1973,Haghighatshoar17}.

Since we are interested in low-complexity estimators, we approximate
the solution by reducing the constraint set to positive
semi-definite, circulant matrices $\set M = \set C^+$.
This choice reduces the complexity of the ML estimator significantly~\cite{neumann_low-complexity_2015}.
The reason is that all circulant matrices have the columns of the
DFT matrix $\vF$ as eigenvectors.
That is, we can parametrize the ML estimate as
\begin{equation}\label{eq:circ_param}
    \vC_\vdelta^\text{ML} = \vF\he \diag( \vc_\vdelta^\text{ML} ) \vF
\end{equation}
where $\vc_\vdelta^\text{ML} \in \R^M$ contains the $M$ eigenvalues of $\vC^\text{ML}_\vdelta$.

Incorporating \eqref{eq:circ_param} into the likelihood function~\eqref{eq:likelihood}, we notice that the estimate of the channel covariance matrix can be given in terms of the estimated power spectrum~\cite{dembo_relation_1986,neumann_low-complexity_2015}
\begin{equation}\label{eq:sample_spectrum}
    \vs = \frac{1}{T} \sum_{t=1}^T \abs{\vF \vy_t}^2
\end{equation}
where $\abs{\vx}^2$ is the vector of absolute squared entries of $\vx$.
Specifically, we have the estimated eigenvalues
\begin{equation}
    \vc_\vdelta^\text{ML} = [\vs - \sigma^2 \ones]_+
\end{equation}
where the $i$th element of $[\vx]_+$ is $\max([\vx]_i,0)$ and where
$\ones$ is the all-ones vector.
The approximate MMSE estimate of the channel vector in coherence
interval $t$ is given by
\begin{equation}\label{eq:circulant_ml}
    \hest_t = \vF\he \diag( \vc_\vdelta^\text{ML} ) \diag( \vc_\vdelta^\text{ML} + \sigma^2\ones)\inv \vF \vy_t
\end{equation}
and can be calculated with a complexity of $\order(M \log M)$ due to the FFT.
The almost linear complexity makes the ML approach with the circulant approximation suitable for large-scale wireless systems.

\subsection{Compressive Sensing Based Estimation}
\label{sec:omp}

The ML-based channel estimation techniques exploit the Toeplitz structure of the covariance matrix, which is a result of regular array geometries and the model~\eqref{eq:covmodel} with a continuous power density function $g$.
In the 3GPP models, this power density function usually has a very limited angular support, i.e., $g(\theta,\vdelta)$ is approximately zero except for $\theta$ in the vicinity of the cluster centers $\vdelta$.
The resulting covariance matrices have a very low numerical rank~\cite{Wiese16}.
As a consequence, under such a model, any given realization of a channel vector admits a sparse approximation
\begin{equation}
\vh \approx \vD\vx
\end{equation}
in a given dictionary $\vD\in\C^{M\times Q}$, where all but $k$ entries of $\vx$ are zero.
The vector $\vx$ can be found by solving the sparse approximation problem
\begin{equation}
\vx = \argmin_{\vx\in\C^Q : |\supp(\vx)| \leq k} \|\vy - \vD\vx\|^2
\end{equation}
where $|\supp(\vx)|$ denotes the number of nonzero entries of $\vx$.
This combinatorial optimization problem can be solved efficiently with methods from the area of compressive sensing, e.g., the orthogonal matching pursuit (OMP) algorithm~\cite{Gharavi98} or iterative hard thresholding (IHT)~\cite{Blumensath09b}.

It is common to use a dictionary $\vD$ of steering vectors $\va(\vtheta)$ where $\vtheta$ varies between $-\pi$ and $\pi$ on a grid~\cite{Heath14}.
For ULAs, this grid can be chosen such that the dictionary $\vD$ results in an oversampled DFT matrix, which has the advantage that matrix-vector products with this matrix can be computed efficiently.
Furthermore, it was shown in~\cite{Wiese16} that this dictionary is a reasonable choice, at least for the single-cluster 3GPP model, and if the OMP algorithm is used to find the sparse approximation.

The OMP algorithm can be extended to a multiple measurement model 
\begin{equation}
    \vH\approx\vD\vX
\end{equation}
with a row-sparse matrix $\vX$, i.e., each channel realization is approximated as a linear combination of the same $k$ dictionary vectors.
Because the selection of the optimal sparsity level $k$ is non-trivial, we use a genie-aided approach in our simulations. 
The genie-aided OMP algorithm uses the actual channel realizations $\vH$ to decide about the optimal value for $k$ that maximizes the metric of interest.
The result is clearly an upper bound for the performance of the OMP algorithm.

\section{Simulations}\label{sec:simulations}

For the numerical evaluation of the newly introduced algorithms, we focus on the far-field model with a ULA at the base station~(cf.~Appendix~\ref{app:ula}).
We assume that the noise variance $\sigma^2$ and the correct model for the parameters, i.e., the prior $p(\vdelta)$ and the mapping from $\vdelta$ to $\vC_\vdelta$, are known.
That is, for the off-line learning procedure required by the CNN estimators, we can use the true prior to generate the necessary realizations of channel vectors and observations.

We first consider the single-path model that motivates Assumption~\ref{assumption:shift_invariance} (cf.~App.~\ref{app:shift_invariance}).
Even for this idealized model, Assumptions 1--3 only hold approximately.
To compare the simple gridded estimator (GE) $\Wgrid$ (Assumption 1) with the structured estimator (SE) $\wstruct$ (Assumptions 1 and 2) and the fast estimator (FE) $\wshift$ (Assumptions 1--3), we first generate $N=16M$ samples $\delta_i\in[-\pi,\pi]$ according to a uniform distribution (single-path model).
We then evaluate the covariance matrices $\vC_{\delta_i}$ and the MMSE filters $\vW_{\delta_i}$ according to~\eqref{eq:covmodel} and~\eqref{eq:cond_mmse_estimator2} with a Laplace power density~\eqref{eq:laplace_power_density} with an angular spread of $\sigma_\text{AS} = 2^{\circ}$.

The gridded estimator $\Wgrid(\cdot)$ is then given by~\eqref{eq:mmseestimator_grid}.
We have chosen $N$ sufficiently large, such that, for the single-path model, the performance of the GE is close to the performance of the (non-gridded) MMSE estimator.
For the SE that uses $\wstruct(\cdot)$ we consider circulant and Toeplitz structure, i.e., we use the DFT matrix $\vQ=\vF$ for the circulant SE and the partial DFT matrix $\vQ=\vF_2$ for the Toeplitz SE as explained in the examples at the end of Sec.~\ref{sec:structured_mmse}.
The coefficients $\vw_{\delta_i}$ in~\eqref{eq:diag_structure} are found by solving a least-squares problem and the respective estimators can then be evaluated as specified in Theorem~\ref{thm}.
Since we have a finite number of antennas, we expect a performance loss compared to $\Wgrid$, due to violation of Assumption~\ref{assumption:diag_structure}.
For the fast estimator that uses $\wshift$ in~\eqref{eq:fast_opt_filt}, we use a circulant structure $\vQ = \vF$ and we only need to calculate $\vw_0$ for $\delta=0$ since the matrices $\vA_\text{SE}$ in $\wstruct$ are replaced by circulant convolutions.

As a baseline, we also show the MSE for the genie-aided MMSE estimator, which simply uses $\vW_\delta$ for the correct $\delta$.
The per-antenna MSE of the channel estimation for the different approximations for a single snapshot ($T=1$) is depicted in Fig.~\ref{fig:wrt_M_single_cluster} as a function of the number of antennas $M$ for a fixed SNR of \SI{0}{\decibel}.
We see, indeed, a gap between the gridded estimator and the two structured estimators.
As expected, the Toeplitz SE outperforms the circulant SE. 
For this scenario, the FE yields performance close to the circulant SE.
Apparently, the assumption of shift invariance is reasonably accurate.
For a large number of antennas, the relative difference in performance of the algorithms diminishes and all algorithms get quite close to the genie-aided estimator.

\newlength{\plotheight}
\setlength{\plotheight}{0.9\columnwidth}
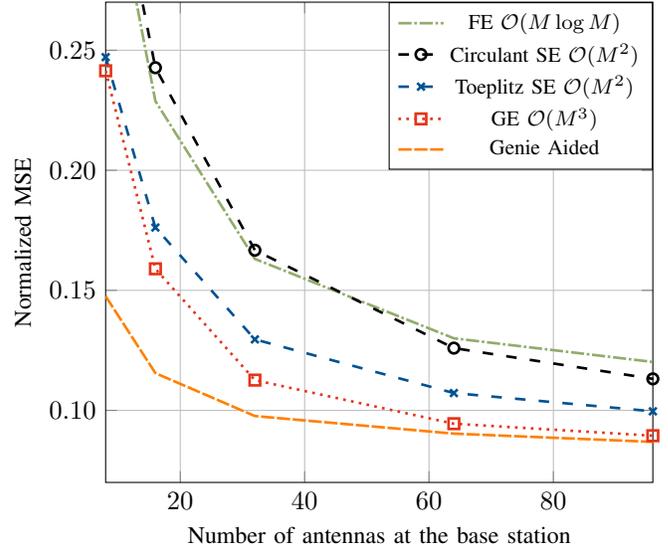
\begin{figure}[t]
    \centering
    \begin{tikzpicture}%[trim axis left, trim axis right]
\begin{axis}[
    height=\plotheight,
    width=\columnwidth,
    grid=both,
    ticks=both,
    xmin = 8, xmax = 96,
    ymin = 0.07, ymax = 0.27,
    ylabel={\small Normalized MSE },
    xlabel={\small Number of antennas at the base station},
    title={},
    yticklabel style={/pgf/number format/.cd, fixed, fixed zerofill, precision=2},
    legend entries={
        \footnotesize FE $\order(M\log M)$,
        \footnotesize Circulant SE $\order(M^2)$,
        \footnotesize Toeplitz SE $\order(M^2)$,
        \footnotesize GE $\order(M^3)$,
        \footnotesize Genie Aided,
    },
    legend style={at={(1.0,1.0)}, anchor=north east},
]
\addplot[fe,discard if not={Algorithm}{FastMMSE} ]
    table [ignore chars=",x=nAntennas,y=MSE,col sep=comma]
    {figures/figure5.csv};
\addlegendentry{\footnotesize FE $\order(M\log M)$};

\addplot[circse,discard if not={Algorithm}{CircMMSE} ]
    table [ignore chars=",x=nAntennas,y=MSE,col sep=comma]
    {figures/figure5.csv};
\addlegendentry{\footnotesize Circulant SE $\order(M^2)$};

\addplot[toepse,discard if not={Algorithm}{ToepMMSE} ]
    table [ignore chars=",x=nAntennas,y=MSE,col sep=comma]
    {figures/figure5.csv};
\addlegendentry{\footnotesize Toeplitz SE $\order(M^2)$};

\addplot[gridest,discard if not={Algorithm}{DiscreteMMSE} ]
    table [ignore chars=",x=nAntennas,y=MSE,col sep=comma]
    {figures/figure5.csv};
\addlegendentry{\footnotesize GE $\order(M^3)$};

\addplot[genie,discard if not={Algorithm}{GenieMMSE} ]
    table [ignore chars=",x=nAntennas,y=MSE,col sep=comma]
    {figures/figure5.csv};
\addlegendentry{\footnotesize Genie Aided};
\end{axis}
\end{tikzpicture}
\caption{
    MSE per antenna at an SNR of 0\,dB for estimation from a single snapshot ($T=1$). Channel model with one propagation path with uniformly distributed angle and a per path angular spread of $\sigma_\text{AS} = 2^{\circ}$.
}
    \label{fig:wrt_M_single_cluster}
\end{figure}

We see that for this simple channel model, there is not much potential for our learning-based methods.
However, the results change significantly for a more realistic channel model.
In the following, we consider results for the 3GPP model with three propagation paths, which have different relative path gains.
That is, the power density is given by 
\begin{equation}\label{eq:channel_model_3p}
    g_\text{3p}(\theta, \vdelta = [\delta_1,\delta_2,\delta_3,p_1,p_2,p_3]\tp ) = \sum_{i=1}^3 p_i g_\text{lp}(\theta, \delta_i)
\end{equation}
where the angles $\delta_i$ are uniformly distributed.
The path gains $p_i$ are drawn from a uniform distribution in the interval $[0,1]$ and then normalized such that $\sum_i p_i = 1$.
The angular spread of each path is still $\sigma_\text{AS} = 2^{\circ}$.

In Figs.~\ref{fig:wrt_M_three_clusters} and~\ref{fig:wrt_SNR_three_clusters}, we show the resulting normalized MSE for the numerical simulation with three propagation paths.
We see that the gap between the fast estimator and the Toeplitz SE is much larger than in Fig.~\ref{fig:wrt_M_single_cluster}.
The fast estimator does not perform well in this scenario as the shift-invariance assumption is lost when the model contains more than one propagation path.

\begin{figure}[t]
    \centering
    \begin{tikzpicture}
\begin{axis}[
    height=\plotheight,
    width=\columnwidth,
    grid=both,
    ticks=both,
    xmin = 8, xmax = 128,
    ymin = 0.1, ymax = 0.72,
    ylabel={\small Normalized MSE },
    xlabel={\small Number of antennas at the base station},
    title={},
    yticklabel style={/pgf/number format/.cd, fixed, fixed zerofill, precision=1},
    legend style={at={(1.00,1.00)}, anchor=north east},
]
\addplot[omp,discard if not={Algorithm}{GenieOMP} ]
    table [ignore chars=",x=nAntennas,y=MSE,col sep=comma]
    {figures/figure6.csv};
\addlegendentry{\footnotesize Genie OMP $\order(M\log M)$};

\addplot[ml,discard if not={Algorithm}{CircML} ]
     table [ignore chars=",x=nAntennas,y=MSE,col sep=comma]
     {figures/figure6.csv};
\addlegendentry{\footnotesize ML~\eqref{eq:circulant_ml} $\order(M\log M)$};

\addplot[fe,discard if not={Algorithm}{FastMMSE} ]
    table [ignore chars=",x=nAntennas,y=MSE,col sep=comma]
    {figures/figure6.csv};
\addlegendentry{\footnotesize FE $\order(M\log M)$,}

\addplot[softmax,discard if not={Algorithm}{CircSoftmax} ]
    table [ignore chars=",x=nAntennas,y=MSE,col sep=comma]
    {figures/figure6.csv};
\addlegendentry{\footnotesize Softmax $\order(M\log M)$};

\addplot[relu,discard if not={Algorithm}{ToepReLU} ]
    table [ignore chars=",x=nAntennas,y=MSE,col sep=comma]
    {figures/figure6.csv};
\addlegendentry{\footnotesize ReLU $\order(M \log M)$};

\addplot[toepse,discard if not={Algorithm}{ToepMMSE} ]
    table [ignore chars=",x=nAntennas,y=MSE,col sep=comma]
    {figures/figure6.csv};
\addlegendentry{\footnotesize Toeplitz SE $\order(M^2)$};

\addplot[genie,discard if not={Algorithm}{GenieMMSE} ]
    table [ignore chars=",x=nAntennas,y=MSE,col sep=comma]
    {figures/figure6.csv};
\addlegendentry{\footnotesize Genie Aided};
\end{axis}
\end{tikzpicture}
\caption{
    MSE per antenna at an SNR of 0\,dB for estimation from a single snapshot ($T=1$). Channel model with three propagation paths (cf.~\eqref{eq:channel_model_3p}).
}
    \label{fig:wrt_M_three_clusters}
\end{figure}
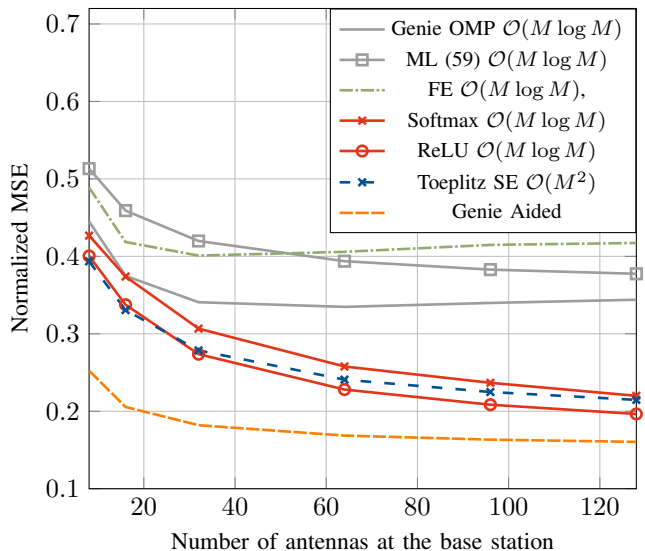

This is where the learning-based estimators shine, since they can potentially compensate for inaccurate assumptions.
We distinguish between the CNN estimator using the softmax activation function and the one with a rectified linear unit (ReLU) as activation function.
In both cases, we only show results for $\vQ=\vF_2$, since using $\vQ=\vF$ lead to consistently worse results.
We ran the hierarchical learning procedure described in Sec.~\ref{sec:learning} for $10\,000$ iterations with mini-batches of 20 samples generated from the channel model.
We also include results for the ML estimator and the genie-aided OMP algorithm discussed in Sections~\ref{sec:ml} and~\ref{sec:omp}, respectively. 
For the OMP algorithm we use a four-times oversampled DFT matrix as dictionary $\vD$.

The performance of the softmax-CNN estimator shows that it is, indeed, a good idea to use optimized variables instead of plug-in values that were derived under assumptions that fail to hold.
It is astonishing that the ReLU-CNN estimator, which has the same computational complexity as the softmax-CNN estimator, significantly outperforms all other estimators of comparable complexity.
In fact, the ReLU-CNN estimator even outperforms the more complex Toeplitz SE estimator.
This can be explained by the fact that, compared to the single-path model, the number of parameters $\vdelta$ is increased and the choice of $N=16M$ samples no longer guarantees a small gridding error.

\begin{figure}[t]
    \centering
    \begin{tikzpicture}
\begin{semilogyaxis}[
    height=\plotheight,
    width=\columnwidth,
    grid=both,
    ticks=both,
    xmin = -15, xmax = 15,
    ymin = 0.01, ymax = 15.0,
    ylabel={\small Normalized MSE },
    xlabel={\small SNR [dB]},
    title={},
    legend style={at={(1.0,1.0)}, anchor=north east},
]
\addplot[omp,discard if not={Algorithm}{GenieOMP} ]
    table [ignore chars=",x=SNR,y=MSE,col sep=comma]
    {figures/figure7.csv};
\addlegendentry{\footnotesize Genie OMP $\order(M\log M)$};

\addplot[ml,discard if not={Algorithm}{CircML} ]
     table [ignore chars=",x=SNR,y=MSE,col sep=comma]
     {figures/figure7.csv};
\addlegendentry{\footnotesize ML~\eqref{eq:circulant_ml} $\order(M\log M)$};

\addplot[fe,discard if not={Algorithm}{FastMMSE} ]
    table [ignore chars=",x=SNR,y=MSE,col sep=comma]
    {figures/figure7.csv};
\addlegendentry{\footnotesize FE $\order(M\log M)$,}

\addplot[softmax,discard if not={Algorithm}{CircSoftmax} ]
    table [ignore chars=",x=SNR,y=MSE,col sep=comma]
    {figures/figure7.csv};
\addlegendentry{\footnotesize Softmax $\order(M\log M)$};

\addplot[relu,discard if not={Algorithm}{ToepReLU} ]
    table [ignore chars=",x=SNR,y=MSE,col sep=comma]
    {figures/figure7.csv};
\addlegendentry{\footnotesize ReLU $\order(M \log M)$};

\addplot[toepse,discard if not={Algorithm}{ToepMMSE} ]
    table [ignore chars=",x=SNR,y=MSE,col sep=comma]
    {figures/figure7.csv};
\addlegendentry{\footnotesize Toeplitz SE $\order(M^2)$};

\addplot[genie,discard if not={Algorithm}{GenieMMSE} ]
    table [ignore chars=",x=SNR,y=MSE,col sep=comma]
    {figures/figure7.csv};
\addlegendentry{\footnotesize Genie Aided};
\end{semilogyaxis}
\end{tikzpicture}
    \caption{MSE per antenna for M=64 antennas and for estimation from a single snapshot ($T=1$). Channel model with three propagation paths (cf.~\eqref{eq:channel_model_3p}).}
    \label{fig:wrt_SNR_three_clusters}
\end{figure}
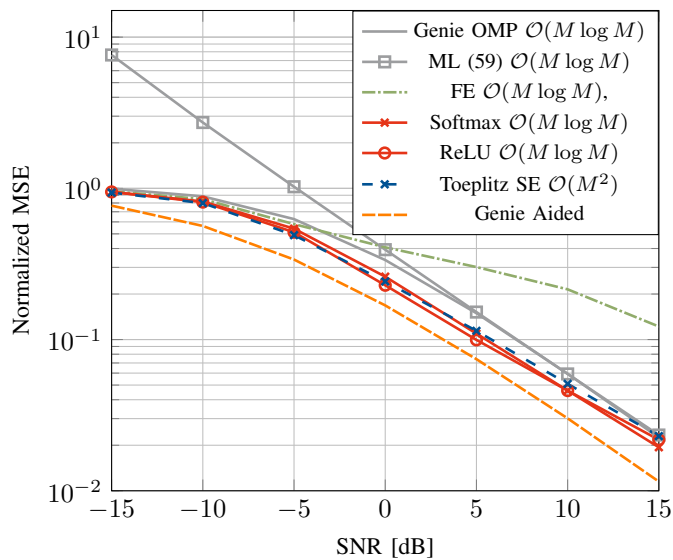

\begin{table}
\caption{Simulation parameters for Fig.~\ref{fig:wrtT}}
\label{tab:sim_parameters}
\begin{center}
\begin{tabular}{lr} \toprule
    %A & B \\ \midrule
    Path-loss coefficient & \num{3.5} \\
    Log-normal shadow fading & \SI{0}{\decibel} \\
    Min. distance & \SI{1000}{m} \\
    Max. distance & \SI{1500}{m} \\ 
    SNR at max. distance & \SI{-10}{\decibel} \\ \bottomrule
\end{tabular}
\end{center}
\end{table}

Finally, we use the 3GPP urban-macro channel model as specified in~\cite{3gpp} with a single user placed at different positions in the cell.
The parameters used in the simulation are given in Table~\ref{tab:sim_parameters}.
In Fig.~\ref{fig:wrtT}, we depict the performance in terms of spectral efficiency with respect to the number of available observations.
Specifically, we use a matched filter in the uplink and evaluate the rate expression
\begin{align}\label{eq:rate}
    r = \expec \left[ \log_2 \left( 1 + \frac{\lvert \hest\he\vh \rvert^2}{\sigma^2 \lVert\hest\rVert^2} \right)\right]
\end{align}
with Monte Carlo simulations.
We assume that the SNR is the same during training and data transmission.
Note that this rate expression, which assumes perfect channel state information (CSI) at the decoder, yields an upper bound on the achievable rate, which is a simple measure for the accuracy of the estimated subspace.
Commonly used lower bounds on the achievable rate that take the imperfect CSI into account are not straightforward to apply to our system model and are, therefore, not shown.

The performance of both the ML estimator and the ReLU-CNN estimator converge towards the genie aided estimator for large $T$.
For a small to moderate number of observations, the CNN-based approach is clearly superior.
The upper bound on the OMP performance is still lower than the performance of the circulant ML estimator.

We also see that the ReLU-CNN estimator outperforms the Toeplitz SE for a high number of observations.
The reason is that we use a fixed number of samples $N=16M$, which leads to an error floor with respect to the number of observations.
In other words, for higher numbers of observations the complexity of the gridded estimator has to be increased to improve the estimation accuracy.
In contrast, the accuracy of the ReLU-CNN estimator improves just like that.

The simulation code is available online~\cite{neumann_simulation_2017}.

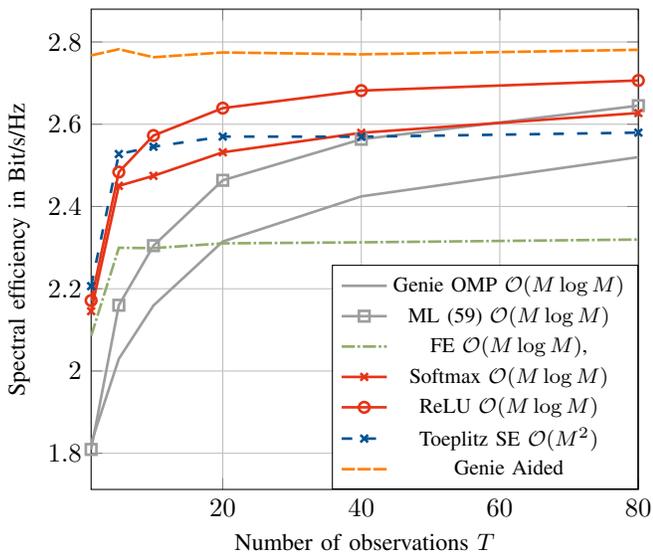
\begin{figure}[t]
    \centering
    \begin{tikzpicture}
\begin{axis}[
    height=\plotheight,
    width=\columnwidth,
    grid=both,
    ticks=both,
    xmin = 1, xmax = 80,
    ylabel={\small Spectral efficiency in Bit/s/Hz},
    xlabel={\small Number of observations $T$},
    title={},
    legend entries={
        \footnotesize Genie OMP $\order(M\log M)$,
        \footnotesize ML~\eqref{eq:circulant_ml} $\order(M\log M)$,
        \footnotesize FE $\order(M\log M)$,
        \footnotesize Softmax $\order(M\log M)$,
        \footnotesize ReLU $\order(M \log M)$,
        \footnotesize Toeplitz SE $\order(M^2)$,
        \footnotesize Genie Aided,
    },
    legend style={at={(1.0,0.0)}, anchor=south east},
]
\addplot[omp,discard if not={Algorithm}{GenieOMP} ]
    table [ignore chars=",x=nCoherence,y=rate,col sep=comma]
    {figures/figure8.csv};
\addlegendentry{\footnotesize Genie OMP $\order(M\log M)$};
\addplot[ml,discard if not={Algorithm}{CircML} ]
     table [ignore chars=",x=nCoherence,y=rate,col sep=comma]
     {figures/figure8.csv};
\addlegendentry{\footnotesize ML~\eqref{eq:circulant_ml} $\order(M\log M)$};
\addplot[fe,discard if not={Algorithm}{FastMMSE} ]
    table [ignore chars=",x=nCoherence,y=rate,col sep=comma]
    {figures/figure8.csv};
\addlegendentry{\footnotesize FE $\order(M\log M)$,}
\addplot[softmax,discard if not={Algorithm}{CircSoftmax} ]
    table [ignore chars=",x=nCoherence,y=rate,col sep=comma]
    {figures/figure8.csv};
\addlegendentry{\footnotesize Softmax $\order(M\log M)$};
\addplot[relu,discard if not={Algorithm}{ToepReLU} ]
    table [ignore chars=",x=nCoherence,y=rate,col sep=comma]
    {figures/figure8.csv};
\addlegendentry{\footnotesize ReLU $\order(M \log M)$};
\addplot[toepse,discard if not={Algorithm}{ToepMMSE} ]
    table [ignore chars=",x=nCoherence,y=rate,col sep=comma]
    {figures/figure8.csv};
\addlegendentry{\footnotesize Toeplitz SE $\order(M^2)$};
\addplot[genie,discard if not={Algorithm}{GenieMMSE} ]
    table [ignore chars=",x=nCoherence,y=rate,col sep=comma]
    {figures/figure8.csv};
\addlegendentry{\footnotesize Genie Aided};
\end{axis}
\end{tikzpicture}
    \caption{
        Spectral efficiency for $M=64$ antennas and an SNR of \SI{-10}{dB} at the cell edge. 
        The urban macro channel model specified in~\cite{3gpp} is used to generate the channels.
    }
    \label{fig:wrtT}
\end{figure}

\section{Conclusion}

We presented a novel approach to learn a low-complexity channel estimator, which is motivated by the structure of the MMSE estimator.
In contrast to other approaches, there are no model parameters which have to be fine-tuned for different channel models.
These parameters are learned from channel realizations that could be generated from a model or measured.
Despite this lack of explicit fine tuning, the proposed method outperforms state-of-the-art approaches at a very low computational cost.
Although we could consider more general NNs, e.g., by replacing convolution matrices with arbitrary matrices, our simulation results suggest that this is not worthwhile, at least as long as the 3GPP models are used.

It will be interesting to establish whether the NN estimators perform equally well for channel models in which the Toeplitz assumption is not satisfied.
In fact, recent work~\cite{gao_massive_2015} suggests that the model in~\eqref{eq:covmodel} based on the far-field assumption does not provide a perfect fit when using large arrays with lots of antennas.
However, the structure of the neural network is not required to perfectly fit the channel model, since the optimized variables can compensate for an inappropriate structure, at least partially.
The only requirement is that suitable training data for the learning procedure is available.

\appendix

\subsection{Proof of Lemma~\ref{lem:mmseestimator}}\label{app:lemma1}

We show Lemma~\ref{lem:mmseestimator} for the slightly more general
case with arbitrary full-rank noise covariance matrices $\cz$.
Let $\vS = T^{-1}\sum_{t=1}^T \vy\vy\he$ denote the sample covariance matrix.
The likelihood of $\vY$ in
\begin{equation}
\hat{\vh}_\text{MMSE}
= \frac{\int p(\vY | \vdelta) \vW_\vdelta\,p(\vdelta)d\vdelta}{\int p(\vY|\vdelta)\,p(\vdelta)d\vdelta}\vy
\end{equation}
is proportional to (we only need to consider factors with $\vdelta$, because other terms cancel out)
\begin{align}
p(\vY | \vdelta )
&\propto \prod_{t=1}^T \frac{\exp\big(-\vy_t\he (\vC_\vdelta + \cz)^{-1} \vy_t \big)}{\abs{\vC_\vdelta + \cz} } \\
&= \exp\big(-T \tr((\vC_\vdelta + \cz)^{-1}\vS) \big) \prod_{t=1}^T |(\vC_\vdelta + \cz)^{-1}|.\label{eq:app:likelihood}
\end{align}
We express $(\vC_\vdelta + \cz)^{-1}$ in terms of $\vW_\vdelta$ as
follows:
We have
\begin{align}
    \vC_\vdelta &= \vW_\vdelta(\vC_\vdelta + \cz) \\ 
    \Leftrightarrow\; \vC_\vdelta + \cz &= \vW_\vdelta(\vC_\vdelta + \cz) + \cz \\
    \Leftrightarrow\; \id &= \vW_\vdelta + \cz(\vC_\vdelta + \cz)\inv\\
    \Leftrightarrow\; \cz\inv(\id -\vW_\vdelta) &= (\vC_\vdelta + \cz)\inv.
\end{align}
If we plug this expression into the likelihood~\eqref{eq:app:likelihood}, we obtain
\begin{align}
p(\vY\vert\vdelta)
&\propto \exp\big(-T\tr(\cz^{-1} (\id-\vW_\vdelta) \vS)\big) \prod_{t=1}^T |\cz^{-1}(\id-\vW_\vdelta)|\\
&\propto \exp\big(T\tr(\cz^{-1} \vW_\vdelta \vS)\big)\prod_{t=1}^T |\id - \vW_\vdelta|\\
&=\exp\big( T\tr(\cz^{-1}\vW_\vdelta\vS) + T\log|\id - \vW_\vdelta|\big)
\end{align}
since $\cz$ does depend on $\vdelta$.
If we substitute $\cz = \sigma^2\id$ and $\widehat{\vC} = T/\sigma^2 \vS$, Lemma~\ref{lem:mmseestimator} follows.

\subsection{Uniform Linear Array}
\label{app:ula}

For a uniform linear array (ULA) with half-wavelength spacing at the base station, the steering vector is given by
\begin{equation}
\va(\theta) = \big[1, \; \exp(i\pi\sin\theta), \; \ldots, \; \exp(i\pi(M-1)\sin\theta)\big]\he.
\end{equation}
Consequently, the covariance matrix has Toeplitz structure with entries
\begin{align}
    [\vC_\vdelta]_{mn} = \int_{-\pi}^{\pi} g(\theta;\vdelta) \exp(-i\pi(m-n)\sin\theta) d\theta.
\end{align}
If we substitute $u = \pi \sin \theta$, we get
\begin{align}\label{eq:fourier_trafo}
    [\vC_\vdelta]_{mn} = \frac{1}{2\pi} \int_{-\pi}^{\pi} f(u;\vdelta) \exp(-i(m-n)u) du
\end{align}
with
\begin{equation}\label{eq:app:transformedspectrum}
    f(u;\vdelta) = 2\pi \frac{ g(\arcsin(u/\pi);\vdelta) + g( \pi - \arcsin(u/\pi);\vdelta)}{\sqrt{\pi^2 - u^2}}
\end{equation}
where we extended $g$ periodically beyond the interval $[-\pi,\pi]$.
That is, the entries of the channel covariance matrix are Fourier coefficients of the periodic spectrum $f(u;\vdelta)$.

An interesting property of the Toeplitz covariance matrices is that we can define a circulant matrix $\widetilde \vC_\vdelta$ with the eigenvalues $f(2\pi k/M;\vdelta)$, $k=0,\ldots,M-1$, such that $\widetilde \vC_\vdelta \asymp \vC_\vdelta$~\cite{gray_2006_toeplitz}.
That is, to get the elements of the circulant matrices, we approximate the integral in~\eqref{eq:fourier_trafo} by the summation
\begin{align}\label{eq:discrete_fourier_trafo}
[\widetilde\vC_\vdelta]_{mn} = \frac{1}{M} \sum_{k=0}^{M-1} f(2\pi k/M;\vdelta) e^{-i(m-n)2\pi k/M}.
\end{align}

\subsection{Uniform Rectangular Array}
\label{app:ura}

To work with a two-dimensional array, we need a three-dimensional channel model.
That is, in addition to the azimuth angle $\theta$, we also need an elevation angle $\phi$ to describe a direction of arrival.
Under the far-field assumption, the covariance matrix is given by
\begin{equation}\label{eq:covmodel3d}
    \vC_\vdelta = \int_{-\pi/2}^{\pi/2} \int_{-\pi}^{\pi} g(\theta,\phi;\vdelta)\va(\theta,\phi)\va(\theta,\phi)\he d\theta d\phi.
\end{equation}

For a uniform rectangular array (URA) with half-wavelength spacing at the base station, we have $M= M_H M_V$ antenna elements, where $M_H$ is the number of antennas in the horizontal direction and $M_V$ the number of antennas in the vertical direction. 
The correlation between the antenna element at position $(m, p)$ and the one at $(n, q)$, given the parameters $\vdelta$, is given by
\begin{align}
    &\int\displaylimits_{-\frac{\pi}{2}}^{\frac{\pi}{2}} \int\displaylimits_{-\pi}^{\pi} g(\theta,\phi;\vdelta) e^{i\pi((n-m)\sin\theta + (q-p)\cos\theta\sin\phi)} d\theta d\phi \\
    = &\int\displaylimits_{-\frac{\pi}{2}}^{\frac{\pi}{2}} \int\displaylimits_{-\frac{\pi}{2}}^{\frac{\pi}{2}} \tilde g(\theta,\phi;\vdelta) e^{i\pi((n-m)\sin\theta +(q-p)\cos\theta\sin\phi)} d\theta d\phi 
\end{align}
where 
\begin{align}
    \tilde g(\theta, \phi; \vdelta) = g(\theta,\phi;\vdelta) + g(\pi - \theta,\phi;\vdelta).
\end{align}
We can map the square $[-\pi/2, \pi/2]^2$ bijectively onto the circle with radius $\pi$ with the substitution $u = \pi\sin \theta$ and $\nu = \pi \cos\theta\sin\phi$.
The transformed integral can be written as
\begin{multline}
    \int_{-\pi}^{\pi} \int_{-\pi}^{\pi} f(u,\nu;\vdelta) e^{-i\pi((m-n)u +(p-q)\nu)} du\, d\nu 
\end{multline}
with
\begin{equation}
    f(u,\nu;\vdelta) = \begin{cases}
        \tilde f(u,\nu;\vdelta), & \text{for } u^2 + \nu^2 \leq \pi^2, \\
        0, & \text{otherwise.}
    \end{cases}
\end{equation}
The nonzero entries of the two dimensional spectrum are given by
\begin{align}
    \tilde f(u,\nu;\vdelta) = \frac{\tilde g(\arcsin(u/\pi), \arcsin(\nu/(\pi\sqrt{1-u^2}))}{\sqrt{(\pi^2 -u^2)(\pi^2 - u^2 - \nu^2)}}.
\end{align}
That is, for a URA, the entries of the channel covariance matrix are two-dimensional Fourier coefficients of the periodic spectrum $f(u,\nu;\vdelta)$.

We can use the results for the ULA case to show that the URA covariance matrix is asymptotically equivalent to a nested circulant matrix with the eigenvalues $f(2\pi m/M_H, 2\pi p/M_V; \vdelta)$ where $m=0,$ \dots, $M_H-1$ and $p = 0$, \dots, $M_V-1$.
The eigenvectors of the nested circulant matrix are given by $\vF_{M_H} \otimes \vF_{M_V}$ where $\vF_M$ denotes the $M$-dimensional DFT matrix.

To show the asymptotic equivalence, we first replace the Toeplitz structure along the horizontal direction by a circulant structure.
This yields an asymptotically equivalent matrix due to the results from~\cite{gray_2006_toeplitz}.
Second, we replace the Toeplitz structure along the vertical direction by a circulant structure to get the desired result.
Clearly, the asymptotic equivalence only holds if $M_H$ and $M_V$ both go to infinity.

\subsection{Shift Invariance}
\label{app:shift_invariance}

To get circulant matrices $\vA_\text{SE}$ in the structured estimator $\wstruct$ in~\eqref{eq:diag_opt_filt}, we need several assumptions.
First, we assume that the circulant approximation in~\eqref{eq:discrete_fourier_trafo} holds exactly, i.e., the columns $\vw_i$ of $\vA_\text{SE}$ contain uniform samples of the continuous filter
\begin{align}\label{eq:app:continuousfilter}
    w(u;\vdelta_i) = \frac{f(u; \vdelta_i)}{f(u;\vdelta_i) + \sigma^2}.
\end{align}
Next, we assume a single parameter $\delta$ and shift invariance of the spectrum, i.e., $f(u; \delta) = f(u-\delta)$ from which $w(u;\delta) = w(u-\delta)$ follows.
Finally, the prior of $\delta$ has to be uniform on the same grid that generates the samples of the $\vw_i$.

\textbf{Example.} An example that approximately fulfills these assumptions is the 3GPP spatial channel model for a ULA with only a single propagation path.
In this case, we only have one parameter for the covariance matrix: the angle of the path center $\delta$, which is uniformly distributed.
The power density function of the angle of arrival (cf.~\eqref{eq:covmodel}) is given by the Laplace density
\begin{align}\label{eq:laplace_power_density}
    g_\text{lp}(\theta;\delta) = \exp( - d_{2\pi}(\delta,\theta)/\sigma_\text{AS} )
\end{align}
where $d_{2\pi}(\theta,\delta)$ is the wrap-around distance between $\theta$ and $\delta$ and can be thought of as $|\theta-\delta|$ for most $(\theta,\delta)$ pairs.
In other words, for different $\delta$, the function $g_\text{lp}(\theta;\delta)$ is simply a shifted version of $g_\text{lp}(\theta;0)$, i.e., $g_\text{lp}(\theta; \delta) = g_\text{lp}(\theta-\delta;0)$.

\begin{figure}
    \begin{center}
        \begin{tikzpicture}
            \begin{axis}[
                    ylabel={$w_\text{lp}(u;\delta)$},
                    xlabel={$u$},
                    xtick={-3.14159, -1.5708, 0, 1.5708, 3.14159},
                    xticklabels={$-\pi$, $-\pi/2$, $0$, $\pi/2$, $\pi$},
                    xmin=-3.14159,xmax=3.14159,
                    ymin= 0, ymax=1,
                    width=1.00\columnwidth,
                    height=0.65\columnwidth,
                ]
                \foreach \x in {1,2,...,6} 
                    \addplot+[mark=none,solid] table [x index=0, y index=\x,col sep=comma]{figures/filters.csv};
                \foreach \x in {16,17,...,20} 
                    \addplot+[mark=none,solid] table [x index=0, y index=\x,col sep=comma]{figures/filters.csv};
                \foreach \x in {7,8,...,15} 
                    \addplot+[mark=none,dashed] table [x index=0, y index=\x,col sep=comma]{figures/filters.csv};
                \draw[red] (axis cs: 0,0.9) node{\tiny$\delta = 0$};
            \end{axis}
        \end{tikzpicture}
    \end{center}
    \caption{
        Functions $w_\text{lp}(u;\delta)$ for different $\delta \in [-\pi/2,\pi/2]$ sampled on a uniform grid.
        The peaks of the graphs are at $\pi \sin \delta$.
        The graphs for $\delta \in [-\pi/4,\pi/4]$ are depicted with a dashed line-style.
    }
    \label{fig:spectra}
\end{figure}
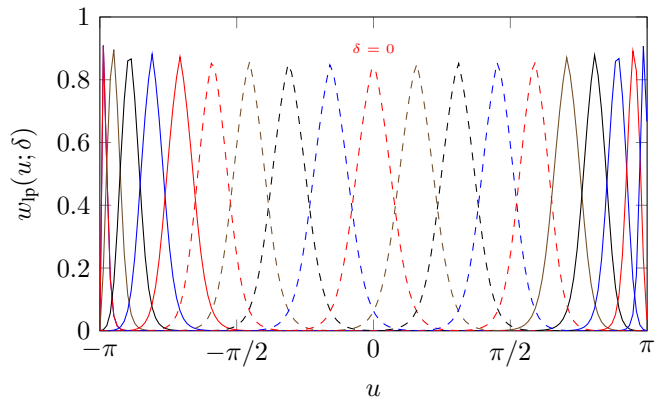

Due to the symmetry of the ULA, we can restrict the parameter $\delta$ to the interval $[-\pi/2,\pi/2]$ without loss of generality.
For angles $\delta \in [-\pi/4,\pi/4]$, i.e., if the cluster center is located at the broadside of the array, the arcsin-transform is approximately linear.
As a consequence, the correspondence~\eqref{eq:laplace_power_density} is approximately true also for the transformed spectrum $f(u;\vdelta)$ (cf.~\eqref{eq:app:transformedspectrum}) and, by virtue of~\eqref{eq:app:continuousfilter}, also the continuous filter is approximately shift-invariant.
This discussion is illustrated by Fig.~\ref{fig:spectra}, which shows the continuous filter $w_\text{lp}(\cdot;\delta)$ for different $\delta$ (the peaks are at $\pi\sin\delta$).
For $\delta\in[-\pi/4,\pi/4]$, the different filters are approximately shifted versions of the central filter, i.e.,
\begin{equation}\label{eq:shift_invariance}
w_\text{lp}(u;\delta) \approx w_\text{lp}(u-\delta; 0).
\end{equation}

For large $M$, the approximation error from using~\eqref{eq:discrete_fourier_trafo} is reduced and we can approximate the matrix $\vA_\text{SE}$ by a circular convolution with uniform samples $\vw_0$ of $w_\text{lp}(u;0)$ as convolution kernel.
We get a \emph{fast estimator} $\wshift$ by setting $\vA_\text{SE}$ in the structured estimator $\wstruct$ to
\begin{align}
    \vA_\text{SE} = \vF\he \diag(\vF \vw_0) \vF
\end{align}
where $[\vw_0]_k = w_\text{lp}(2\pi(k-1)/K; 0)$.

An analogous shift invariance can be derived for a uniform rectangular array.
In this case, we have a two-dimensional shift-invariance and, thus, two-dimensional convolutions.
For the case of distributed antennas that appeared in the examples in Sec.~\ref{sec:structured_mmse}, we do not see a straightforward way to make a similar simplification.

\bibliographystyle{IEEEtran}
\bibliography{IEEEabrv,literature}

\begin{IEEEbiography}[{\includegraphics[width=1in,height=1.25in,clip,keepaspectratio]{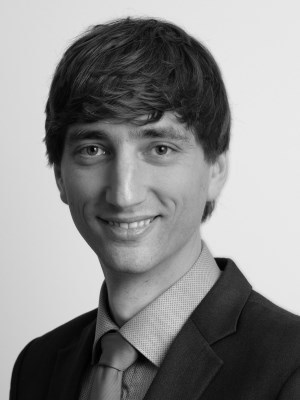}}]{David Neumann}
received the Dipl.-Ing.\ degree in electrical engineering from Technische Universit\"at M\"nchen (TUM) in 2011.
He is currently working towards the doctoral degree at the Professorship for Signal Processing at TUM.
His research interests include transceiver design for large-scale communication systems and estimation theory.
\end{IEEEbiography}%
\begin{IEEEbiography}[{\includegraphics[width=1in,height=1.25in,clip,keepaspectratio]{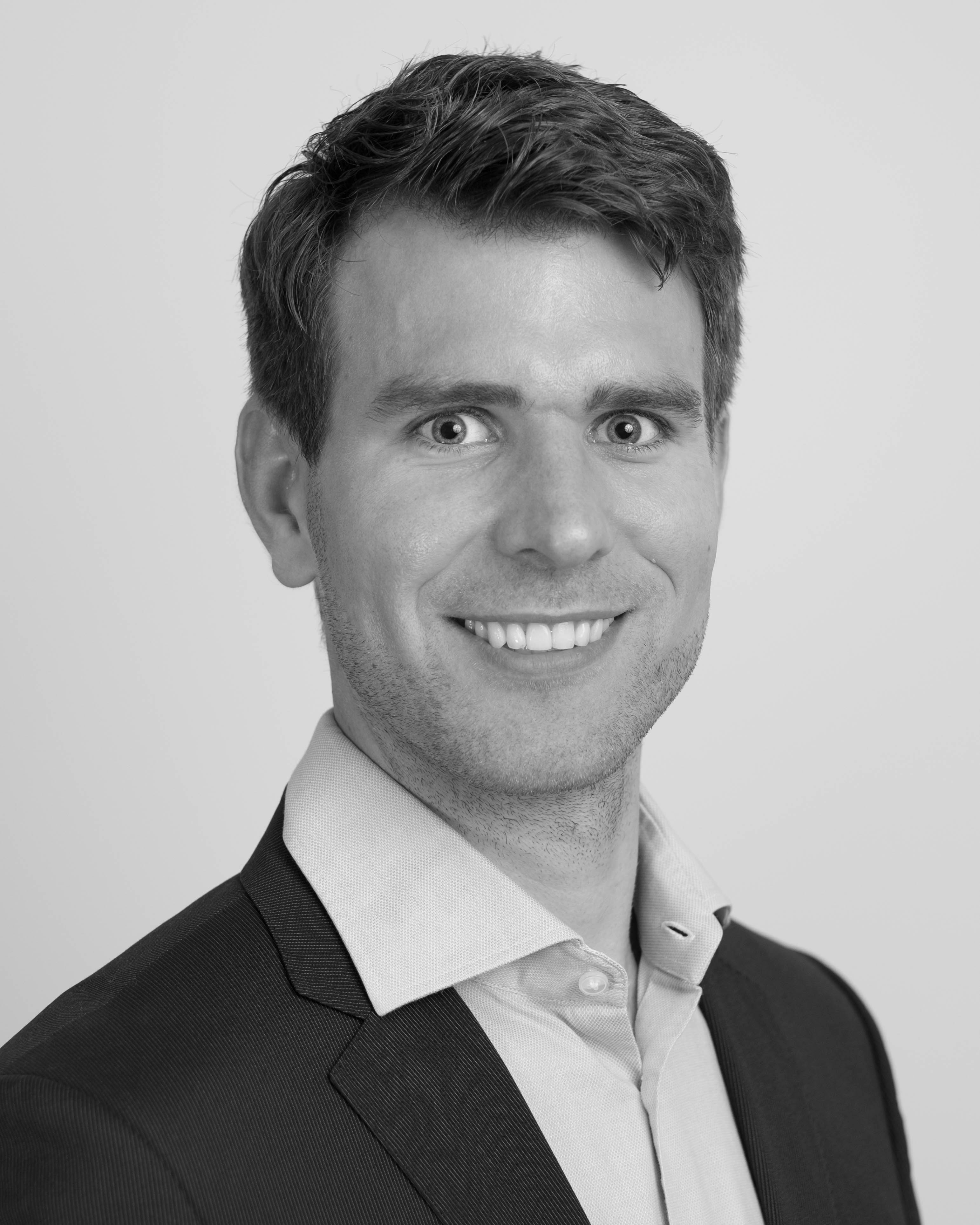}}]{Thomas Wiese}
received the Dipl.-Ing.\ degree in electrical engineering and the Dipl.-Math.\ degree in mathematics from Technische Universit\"at M\"nchen (TUM) in 2011 and 2012, respectively.
He is currently working towards the doctoral degree at the Professorship for Signal Processing at TUM.
His research interests include compressive sensing, sensor array processing, and convex optimization.
\end{IEEEbiography}%
\begin{IEEEbiography}[{\includegraphics[width=1in,height=1.25in,clip,keepaspectratio]{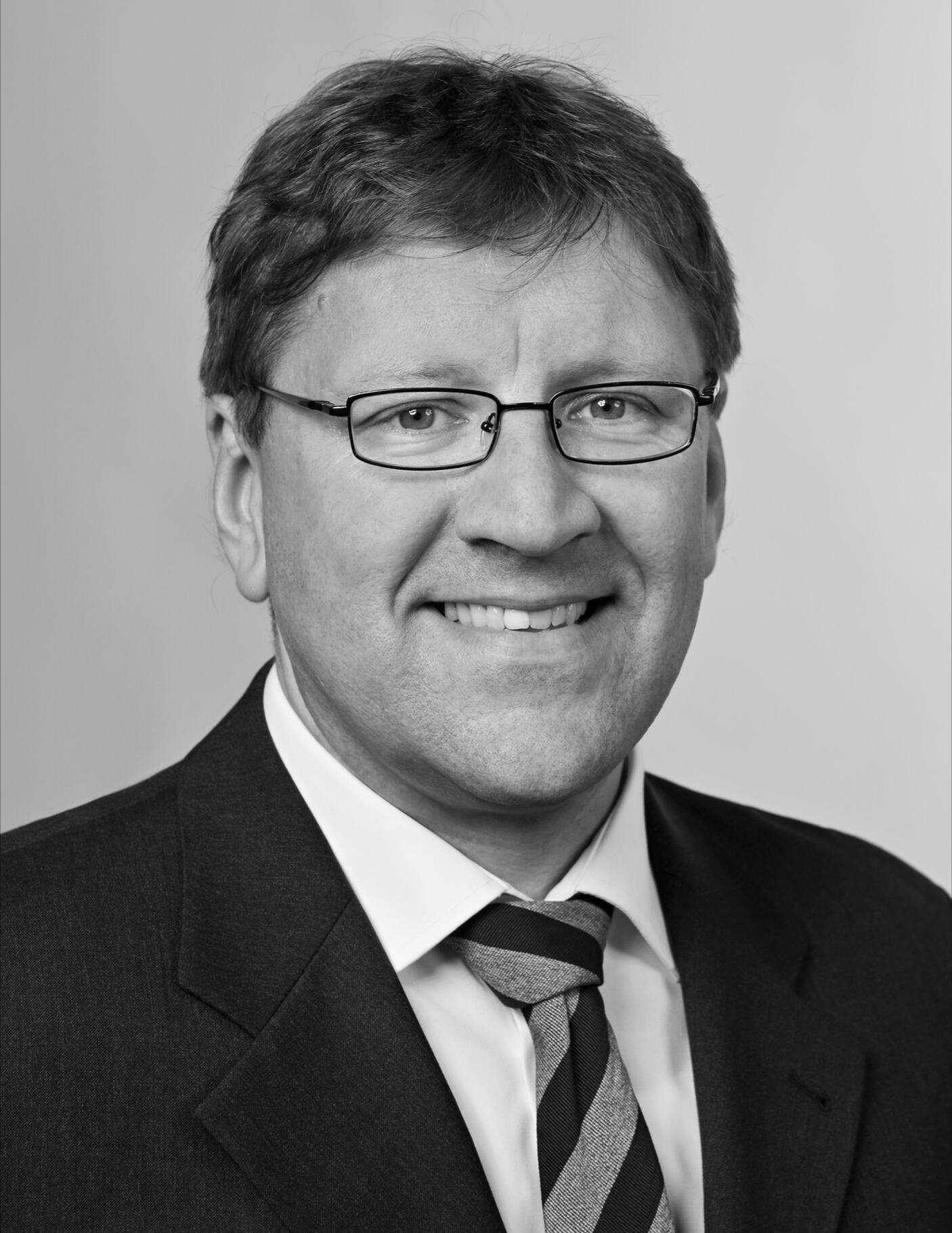}}]{Wolfgang Utschick}
(SM'06) completed several years of industrial training programs before he received the diploma in 1993 and doctoral degree in 1998 in electrical engineering with a dissertation on machine learning, both with honors, from Technische Universit{\"a}t M{\"u}nchen (TUM). Since 2002, he is Professor at TUM where he is chairing the Professorship of Signal Processing. He teaches courses on signal processing, stochastic processes, and optimization theory in the field of wireless communications, various application areas of signal processing, and power transmission systems. Since 2011, he is a regular guest professor at Singapore's new autonomous university, Singapore Institute of Technology (SIT). He holds several patents in the field of multi-antenna signal processing and has authored and coauthored a large number of technical articles in international journals and conference proceedings and has been awarded with a couple of best paper awards. He edited several books and is founder and editor of the Springer book series Foundations in Signal Processing, Communications and Networking. Dr. Utschick has been Principal Investigator in multiple research projects funded by the German Research Fund (DFG) and coordinator of the German DFG priority program Communications Over Interference Limited Networks (COIN). He is a member of the VDE and therein a member of the Expert Group 5.1 for Information and System Theory of the German Information Technology Society. He is senior member of the IEEE, where he is currently chairing the German Signal Processing Section. He has also been serving as an Associate Editor for IEEE Transactions on Signal Processing and has been member of the IEEE Signal Processing Society Technical Committee on Signal Processing for Communications and Networking. Since 2017, he serves as Dean of the Department for Electrical and Computer Engineering at TUM.
\end{IEEEbiography}
\vfill
\end{document}